\documentclass[12pt]{iopart}

\usepackage[utf8]{inputenc}
\usepackage{iopams}
\expandafter\let\csname equation*\endcsname\relax
\expandafter\let\csname endequation*\endcsname\relax
\usepackage{amsmath}
\usepackage{amsthm}
\usepackage{amssymb}
\usepackage{graphicx}
\usepackage{physics}
\usepackage{todonotes}
\usepackage{pgf,tikz,pgfplots}
\usepackage[english]{babel}
\usepackage{tabularx}

\newcounter{protocol}
\newenvironment{protocol}[1]
  {\par\addvspace{\topsep}
   \noindent
   \tabularx{\linewidth}{@{} X @{}}
    \hline
    \refstepcounter{protocol}\textbf{Protocol \theprotocol} #1 \\
    \hline}
  {\\
  \hline
   \endtabularx
   \par\addvspace{\topsep}}

\newtheorem{theorem}{Theorem}
\pgfplotsset{compat=1.14}
\newcommand{\R}{\mathcal{R}}
\newcommand{\mS}{\mathcal{S}}
\newcommand{\MEV}{\mathcal{MEV}}
\newcommand{\MEVfiltered}{\mathcal{MEV}_{C}\bot}
\newcommand{\ConcreteProt}{\pi_{[n]}\mathcal{R}\pi_S}

\newcommand{\Xfamily}{X=\{x_i\}_{i=1}^n}
\newcommand{\pifamily}{\pi=\{\pi_i\}_{i=1}^n}

\begin{document}
\title{Composable Security for Multipartite Entanglement Verification}
\date{January 2020}

\author{R Yehia$^1$, E Diamanti$^1$ and 
I Kerenidis$^2$}

\address{$^1$Sorbonne Université, CNRS, LIP6, F-75005 Paris, France}
\address{$^2$Université de Paris, CNRS, IRIF, F-75013 Paris, France}

\ead{raja.yehia@lip6.fr}

\begin{abstract}
We present a composably secure protocol allowing $n$ parties to test an entanglement generation resource controlled by a possibly dishonest party. The test consists only in local quantum operations and authenticated classical communication once a state is shared among them and provides composable security, namely it can be used as a secure subroutine by $n$ honest parties within larger communication protocols to test if a source is sharing quantum states that are at least $\epsilon$-close to the GHZ state. This claim comes on top of previous results on multipartite entanglement verification where the security was studied in the usual game-based model. Here, we improve the protocol to make it more suitable for practical use in a quantum network and we study its security in the Abstract Cryptography framework to highlight composability issues and avoid hidden assumptions. This framework is a top-to-bottom theory that makes explicit any piece of information that each component (party or resource) gets at every time-step of the protocol. Moreover any security proof, which amounts to showing indistinguishability between an ideal resource having the desired security properties (up to local simulation) and the concrete resource representing the protocol, is composable for free in this setting. This allows us to readily compose our basic protocol in order to create a composably secure multi-round protocol enabling honest parties to obtain a state close to a GHZ state or an abort signal, even in the presence of a noisy or malicious source. Our protocol can typically be used as a subroutine in a Quantum Internet, to securely share a GHZ state among the network before performing a communication or computation protocol.
\end{abstract}


\maketitle

\section{Introduction}


\subsection{Motivation}
Recent conceptual developments in the quantum Internet have allowed to start defining layer models for quantum network architectures \cite{QIavision,IETFlinklayer}. Similar to the OSI model for the classical Internet, these models separate physical and application issues to allow researchers to study experimental problems such as extending coherence time or establishing entangled links between distant nodes using various physical platforms~\cite{loopholefreebell,repeaterNV} and conversely to start developing high-level applications that could sit on top of any physical implementation. A non-exhaustive overview of protocols for the future quantum Internet can be found on the Quantum Protocol Zoo website~\cite{Zoo}. These layer models shed light on composable security issues that have to be addressed. Roughly speaking, a protocol is said to be composably secure if it can be used multiple times in a row or as a subroutine in any bigger protocol without threatening the overall security. A protocol can thus be seen as a black box that can be composed with other protocols, which is precisely the way we would like to think of applications in such settings.

Expected progress within the next few years will lead to several realistic applications such as quantum money~\cite{QuantumMoney,weakcoinflippin}, voting~\cite{Leaderelection}, or anonymous transmission~\cite{Anonymity} that rely on firstly securely establishing entanglement between all nodes in a quantum network. More precisely, many applications can be achieved in such a network of $n$ nodes by first sharing an entangled state then manipulating it locally to get the desired entanglement needed for the rest of the protocol. One such multiparty state is the GHZ state $\frac{\ket{0}^{\bigotimes n}+\ket{1}^{\bigotimes n}}{\sqrt{2}}$ \cite{GHZ} where each party holds one qubit. In this context, it becomes relevant to have a secure protocol for ensuring that the source shares a state that is at least $\epsilon$-close to the GHZ state. In order to be practical, this protocol should have the minimum requirements for the parties, which are to be able to perform classical communication and local single qubit operations. It should also be composable as it is meant to be used as a subroutine in bigger quantum network protocols.

To show composability results, one has to prove security in a composable framework. Here we use the Abstract Cryptography (AC) framework~\cite{AbstractCryptography,ConsCryptography} that captures the ideal-world vs. real-world paradigm. We continue this introduction by presenting the abstract cryptography framework and how composable security can be proven. We try to provide a full introduction to the framework so that it is accessible to non experts in the topic. Then in the second section we present the ideal abstraction of a verification protocol with the desired security properties as well as an actual multipartite entanglement verification protocol that achieves this functionality. This protocol, designed as a subroutine for bigger protocols in a distributed setting, presents interesting security properties and is believed to be achievable in the near future. We study its composability properties in the AC framework. Then in the last section we discuss possible implementations in near-tear quantum networks and limitations of such constructions. 

\subsection{Composability and Abstract Cryptography}

In order to prove composable security, one needs to prove security in a composable framework. One such framework is Abstract (or Constructive) Cryptography (AC), a top-down approach developed by U. Maurer and R. Renner~\cite{AbstractCryptography,FromToConsCrypto,ConsCryptography} to define a simulation-based cryptography theory. It creates some notion of a module with well-defined interfaces that interacts with the rest of the world in a black box fashion. In the Universal Composability framework of Canetti~\cite{UC}, this is called a functionality. In AC, those modules are called resources and going from a resource to another is done through converters called protocols. For example a one-time pad protocol constructs a secure communication channel resource out of a secret key resource and an authenticated classical channel (see Fig.~\ref{fig:OTP}). In their first paper~\cite{AbstractCryptography}, Renner and Maurer defined a complete cryptography algebra of resources with their composition rules. This allowed them to define \textit{equivalence relations} between resources and to infer security notions that inherit composability properties. Moreover this framework is of interest when modeling multiparty protocols as it offers a simpler view of what dishonest parties could have access to than the usual game-based cryptography theory where the strategy for a dishonest group should be given explicitly. The level of abstraction of the different resources can be modulated to highlight the properties that one wants to study about them. Finally, the AC framework is a resource theory with a large power of abstraction that allows us to think of a protocol the same way we would do when thinking of an application in the quantum Internet.

Different results have been achieved using this framework such as the study of unfair coin tossing~\cite{UnfairCoinTossing}, remote state preparation~\cite{RemoteState}, oblivious transfer~\cite{ObliTransfer} and composable security of multiparty delegated quantum computing~\cite{MultipartyQC,DelegatedQC,ComposableMultiDQC}. Different extensions have also been proposed such as adding relativistic constraints~\cite{Relativistic} or global event history in the case of ratcheting~\cite{Ratcheting}. Let us give a brief overview of this framework which we will use to study our multipartite entanglement verification protocol.\\

Abstract cryptography uses the concept of abstract systems to express cryptography as a resource theory. A cryptography protocol is viewed as the construction of some \textit{ideal} resource $\mathcal{S}$ out of other \textit{real} resources $\mathcal{R}$. This construction notion is made through converters. Finally, the distance between two resources is formalised through the notion of a distinguisher. Those three objects are the building blocks of the AC theory. \newpage

\begin{figure}[!ht]
\flushleft
\includegraphics[width=15cm]{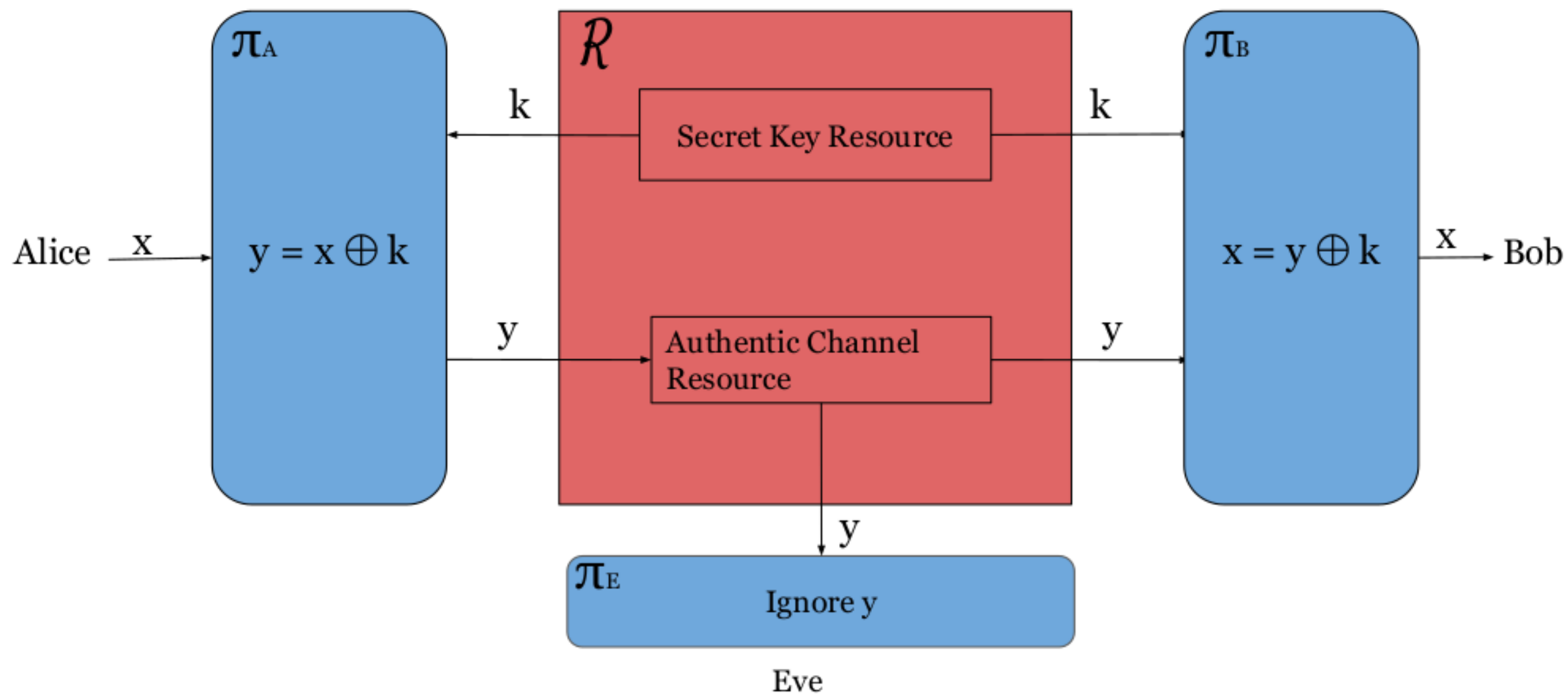}
\caption{Concrete One Time Pad resource $\pi_A\pi_B\R\pi_E$: Alice has access to the left interface, Bob to the right interface and Eve to the down interface. $\R$ is the resource composed of a secret key resource and an authentic channel resource in parallel. Protocols are represented in blue, $\pi_E$ being the protocol of an honest Eve that blocks the input $y$ from the authenticated channel resource.}
\label{fig:OTP}
\end{figure}

A \textbf{resource} is an abstract system with interfaces specified by a set $\mathcal{I}$ (e.g. $\mathcal{I}= \{A,B,E\}$ for Alice, Bob and Eve in a tripartite setting). Each interface is accessible to one user and provides them with some abilities. Note that the notion of a party is not explicitly modeled in this framework, but induced by the interfaces they are restricted to have access to. Resources are used to model functionalities that are not done specifically by a party. They can be associated with real physical resources (e.g. a quantum channel) or with abstract functionalities (e.g. bit commitment or quantum  random number generation). The level of abstraction of such a functionality is not bounded \textit{per se} but it is usually tailored to the application that one is modeling and the properties one wants to highlight. For example quantum memories can be explicit and represented with resources or abstracted in converters. Classical protocols can also be explicitly shown or abstracted through oracle calls. Moreover any parallel composition of resources is a resource in which the interface set corresponds to the union of the ones from the composed resources.

\textbf{Converters} are also abstract systems with one set of ``inside'' interfaces that are expected to be connected to a resource and one set of ``outside'' interfaces. Their name derives from the fact that a converter attached to a resource converts it into another resource by emulating a certain set of interfaces to the outside world. They typically model the local computation of a party during a protocol and are denoted with Greek letters. For a resource $\R$ with interfaces $A$ and $B$ and a two-party protocol $\pi=\{\pi_A,\pi_B\}$ we denote $\pi_A\R\pi_B$ the resource obtained from connecting $\pi_A$ to interface $A$ and $\pi_B$ to interface $B$ (see Fig.~\ref{fig:OTP}). A dishonest party is then modeled by just unplugging their corresponding converter from the resources, indicating that the party is not following the protocol. This leaves the interface they have been accessing open to the outside world.
Note that the ordering of the converters is not important and that they are usually written in the most readable way.

Converters are also used to model the honest utilisation of an ideal resource. Indeed a dishonest user might have access to more functionalities than an honest one. To model this we use a converter, a \textit{filter}, to cover these functionalities for an honest player, that we remove in case of a dishonest utilisation of the resource (see an example in Fig.~\ref{fig:PrivClassChan}). Finally, converters are used in the ideal world to simulate the local output to a dishonest party, in which case we use the term of \textit{simulator}. Converters and resources can be described with the help of boxes and arrows as well as in the form of algorithms by specifying where each output goes.

\begin{figure}[!ht]
    \centering
    \includegraphics[width=9cm]{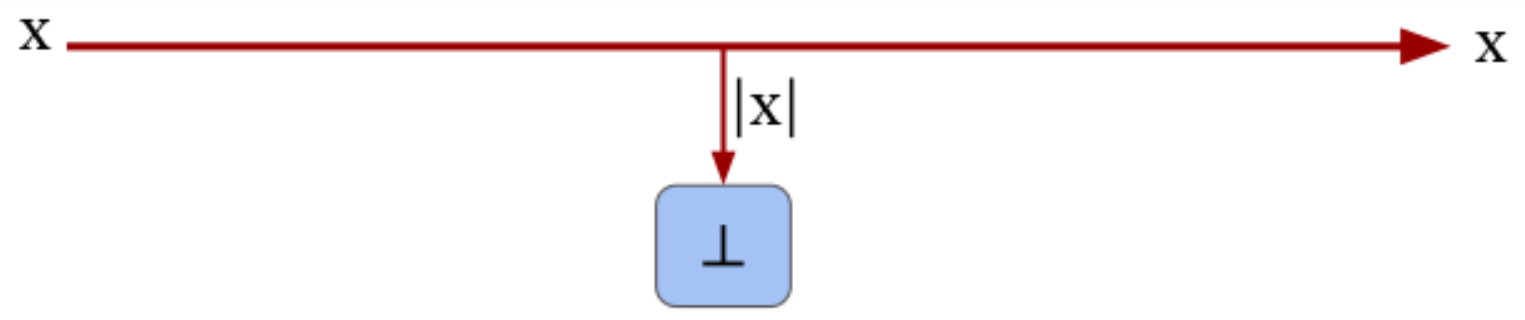}
    \caption{Filtered one-way private classical channel resource. It takes as input a bitstring $x$ at the left interface, outputs it at the right interface and leaks its size $|x|$ on the bottom interface. $\bot$ is a filter blocking the bottom interface to simulate an honest use of the resource. As we will see in the next section, this resource is equivalent to the one of Fig.~\ref{fig:OTP}.}
    \label{fig:PrivClassChan}
\end{figure}

Abstract cryptography is thus the theory of breaking down cryptographic processes into box-shaped resources that can be composed together in series or in parallel. Resources, which represent cryptographic primitives, can be transformed into other resources using converters and composition. A \textit{concrete} resource represents an actual protocol using physical systems and classical and/or quantum operations while an \textit{ideal} resource is the abstraction of the functionality achieved by the protocol. We say that a protocol $\pi = \{\pi_A, \pi_B\}$ constructs the resource $\mS$ out of $\R$ and write $\R\xrightarrow{\pi} \mS$. Such a construction is \textit{composable} if for all $\R,\mS$ and $\mathcal{T}$ resources and $\pi,\nu$ converters (protocols) such that $\R\xrightarrow{\pi} \mS$ and $\mS\xrightarrow{\nu} \mathcal{T}$ we have that

\begin{equation}
    \R \xrightarrow{\pi} \mS \wedge \mS \xrightarrow{\nu} \mathcal{T} \implies \R \xrightarrow{ \nu \circ \pi} \mathcal{T}.
\end{equation}

\subsection{Security definition and assumptions}
\label{subsec:SecDef}
To show that a protocol $\pi$ constructs the ideal $\mS$ out of concrete  resource $\R$, we have to capture an equivalence notion, with a metric $\approx$ such that  $\pi\R \approx \mS \overset{def}{\iff}\R\xrightarrow{\pi} \mS$. To that end, Abstract Cryptography introduces \textbf{Dinstiguishers}. They are abstract systems that are used to construct a pseudo-metric between two resources. They replace the notion of an adversary and also encompass any protocol that is run before, after or during the protocol being analyzed. As its name indicates, a distinguisher is used to distinguish between two resources $\R$ and $\mS$ by connecting to all their interfaces and outputting a single bit: a guess whether it is interacting with $\R$ or $\mS$ (see Fig.~\ref{fig:Distinguisher}). The advantage of a distinguisher $\mathbf{D}$ is given by
\begin{equation*}
    d^{\mathbf{D}}(\R,\mS)=|Pr[\mathbf{D}\R = 0] - Pr[\mathbf{D}\mS = 0]|,
\end{equation*}
where $\mathbf{D}\R$ is the output of $\mathbf{D}$ when interacting with $\R$. For example in Fig.~\ref{fig:Distinguisher}, replacing $\R$ with $\pi_A\pi_B\R\pi_E$ from Fig.\ref{fig:OTP} and $\mS$ with the filtered private authenticated classical channel resource from Fig.~\ref{fig:PrivClassChan}, we see that any distinguisher $\mathbf{D}$ will see the same output $x$ for any given input $x$ on any of the two resources. Hence we have that $d^{\mathbf{D}}(\R,\mS)=0$. For a class of distinguishers $\mathbb{D}$, the distinguishing advantage is defined as
\begin{equation*}
    d^{\mathbb{D}}(\R,\mS) = \underset{\mathbf{D}\in\mathbb{D}}{\sup} d^{\mathbf{D}}(\R,\mS).
\end{equation*}

\begin{figure}[!ht]
    \centering
    \includegraphics[width=13cm]{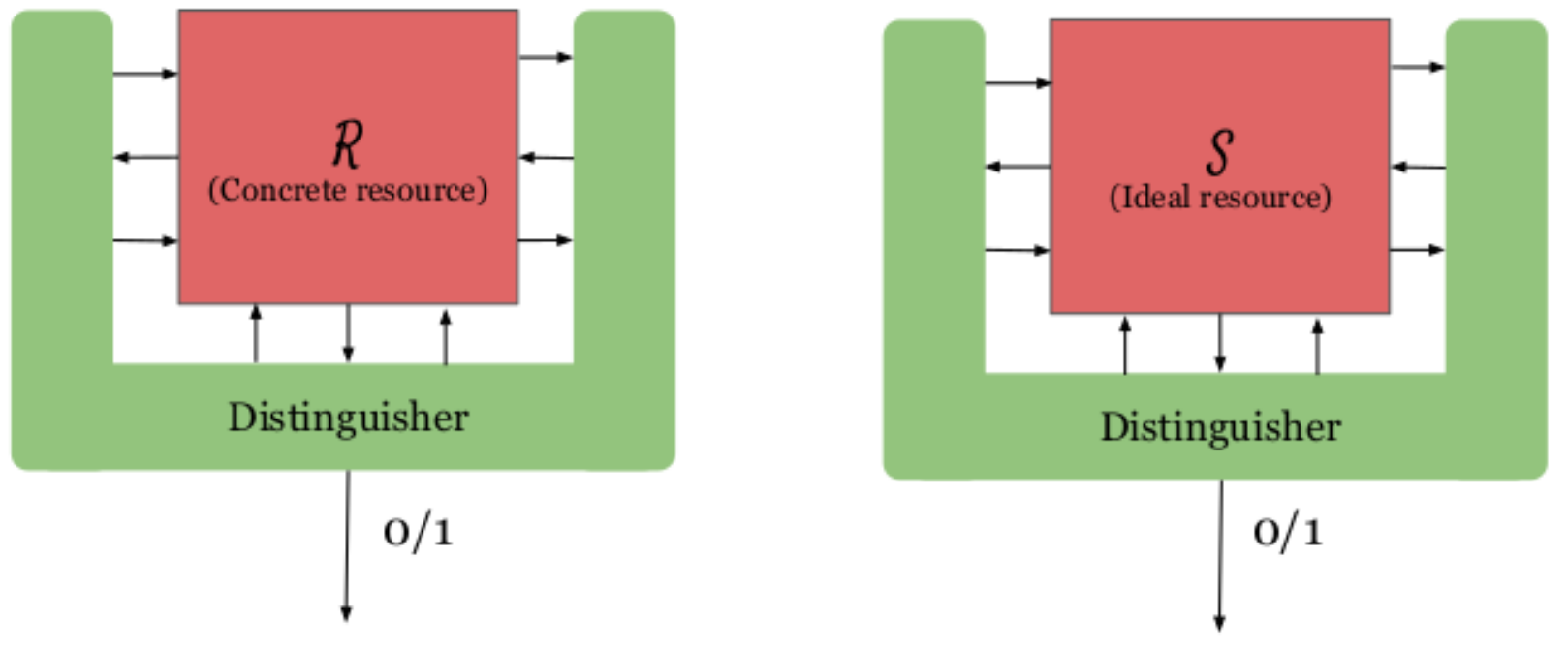}
    \caption{A distinguisher interacting with $\R$ and $\mS$. It has access to a complete description of the two systems and can choose the inputs of all players, receive their outputs and simultaneously fulfill the role of an adversary. After interaction, it must guess which resource is which. Replacing $\R$ by Fig.~\ref{fig:OTP} and $\mS$ by Fig.~\ref{fig:PrivClassChan}, no distinguisher is able to guess between the two resources.}
    \label{fig:Distinguisher}
\end{figure}

The distinguishing advantage is a pseudo-metric on the space satisfying all properties of a composable distance, namely identity, symmetry and triangle inequality. This allows to define \textit{equivalence relations} between resources: for a class of distinguishers $\mathbb{D}$ we say that $\R$ is equivalent (or $\epsilon$-close to) $\mS$ and write $\R \approx \mS$ (resp. $\R \approx_{\epsilon} \mS$) if $d^{\mathbb{D}}(\R,\mS)=0$ (resp. $d^{\mathbb{D}}(\R,\mS)\leq\epsilon$).

To summarize, converters describe mostly local and non-costly operations while resources can have non local functionalities and  extended computational power. Distinguishers are all powerful objects that represent the environment trying to guess between two resources. \\

We now have the necessary ingredients to present the notion of secure construction of a resource in AC. Let $\pi=\{\pi_i\}_{i=1}^n$ be a protocol run by $n$ parties using the concrete resource $\R$ that has interfaces $\mathcal{I}$ and let $\mS$ be an ideal resource with all the desired properties expected from the protocol. We say that \textbf{$\pi$ securely construct $\mS$ out of $\R$  within $\epsilon$} and write $\R\xrightarrow{(\pi,\epsilon)} \mS $ if there exist converters $\sigma = \{\sigma_{i}\}$ called \textit{simulators} such that
\begin{equation}
    \forall \mathcal{P}\subseteq\mathcal{I}, \pi_{\mathcal{P}}\R \approx_{\epsilon} \sigma_{\mathcal{I}\setminus\mathcal{P}}\mS,
\end{equation} 
with $\forall \mathcal{P}\subseteq\mathcal{I}, \pi_{\mathcal{P}}=\{\pi_i\}_{i\in\mathcal{P}}$.

This means that for each party $i$ that does not follow its protocol $\pi_i$, we are able to find a simulator $\sigma_i$ that locally simulates on the ideal resource the interfaces the party has access to on the concrete resource. Simulators don't represent actual concrete operations and should only be seen as a tool in the proof. For example, using a simulator $\sigma$ taking as input a size and producing a random bit string of this size, we have an equivalence relation between the concrete one-time pad resource with an dishonest Eve $\pi_A\pi_B\R$ and the ideal private classical channel resource on which we attach $\sigma$ (see Fig.~\ref{fig:OTPequiv}). This equivalence together with the equivalence of Fig.~\ref{fig:OTP} and Fig.~\ref{fig:PrivClassChan} (usually denoted as the correctness of the protocol) proves the composably secure construction of the private classical channel resource by the one-time pad protocol. In this case, those two equivalences suffice because we suppose Alice and Bob to be always honest in the one-time pad protocol. One must find simulators for each subset of possible dishonest parties to prove composable construction.

\begin{figure}[!ht]
    \centering
    \includegraphics[width = 16cm]{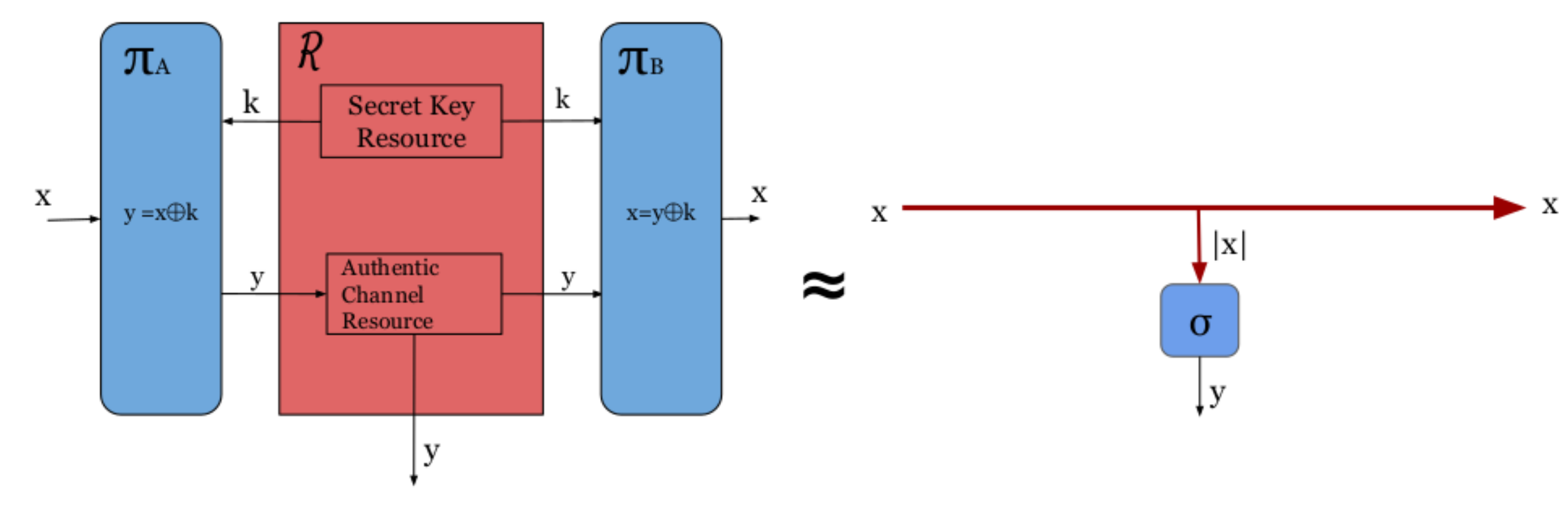}
    \caption{Equivalence between the One-time pad resource with a dishonest Eve and the ideal private classical channel resource with the simulator $\sigma$.}
    \label{fig:OTPequiv}
\end{figure}

The power of the class of distinguishers and simulators used to prove a secure construction determines the strength of the security proof. For example considering only classical distinguishers leads to security against classical adversaries while considering all powerful distinguishers leads to information-theoretic security. Ideally we would want the class of simulators to be restricted to a class of easily implementable converters and the set of distinguishers to be as general as possible. This leads to security statements such as ``We can easily construct the ideal resource $\mS$ from $\R$ and we can easily simulate any cheating behaviour such that even a very powerful distinguisher cannot tell the two resources apart''.
\vspace{1cm}

\section{An entanglement generation testing resource}
In this section we start with a verification protocol that is possible to be realised with current technology and show its equivalence with an ideal verification resource. Then we build upon it to construct an ideal \textit{verified GHZ sharing resource} that $n$ parties of a network can securely use to get a verified GHZ state as a subroutine of a bigger protocol even when the source is noisy or malicious.
\subsection{Multipartite entanglement verification protocol}
\label{subsec:MEVprotocol}
In the following we describe a protocol that securely constructs an ideal multipartite entanglement verification resource using only classical communication between $n$ parties that each receive a single qubit from a source of multipartite entanglement. We believe the following proof can be adapted to any stabilizer state verification where parties first receive a qubit and then do only local operations and classical communication (LOCC).\\ We first review the original protocol from~\cite{MEVresistant}, then introduce the ideal and concrete resources, and finally prove the secure construction. In this paper we will call ``Source'' the party controlling the entanglement source or the device itself interchangeably. We will consider authenticated classical communication and perfect quantum communication as any imperfection can be modeled as the source perfectly sending noisy states.

Our work is based on the work from~\cite{MEVresistant} where the authors develop and analyze an $n$-party verification protocol consisting only of classical communication and local quantum operations once the state is shared. One of the parties, called the \textit{Verifier}, has a central role in the protocol: it sends instructions to all parties and broadcasts the output of the verification. We recall the protocol of~\cite{MEVresistant}:

\begin{protocol}{Multipartite entanglement verification protocol}

\begin{enumerate}
    \item The source creates an $n$-qubit GHZ state and sends each qubit $i$ to party $i$ using a state generation resource and $n$ one-way quantum channels.
    \item The Verifier selects for each $i\in[n]$ a random input $x_{i}\in\{0,1\}$ such that $\sum_{i=1}^{n}x_{i}\equiv 0$ (mod 2) and sends it to the corresponding party via an authenticated classical channel resource. The Verifier keeps one to themselves.
    \item If $x_{i}=0$, party $i$ performs a Hadamard operation on their qubit. If $x_{i}=1$, party $i$ performs a $\sqrt{X}$ operation. 
    \item Each party $i$ measures their qubit in the $\{\ket{0},\ket{1}\}$ basis and sends their outcome $y_{i}$ to the Verifier via the classical channel.
    \item The Verifier accepts and outputs $b_{out}=0$ if and only if 
\begin{equation*}
    \sum_{i=1}^{n}y_{i}\equiv\frac{1}{2}\sum_{i=1}^{n}x_{i} \text{(mod 2)}
\end{equation*}
\end{enumerate}
\end{protocol}
\newpage
This protocol has been extensively studied and presents desirable properties that are expected from such a protocol: it is correct and for one round, its output depends on the distance between the state that was actually shared by the source and the GHZ state and the number of malicious parties. Indeed, for a state $\rho$ shared among the parties, $b_{out}$ is such that:
\begin{equation}
    b_{out} = \left\{
    \begin{array}{ll}
        0 & \mbox{with probability } 1 - \frac{\tau^{2}}{2} \\
        1 & \mbox{with probability } \frac{\tau^{2}}{2}
         
    \end{array}
    \right.
\end{equation}
with 
\begin{equation}
    \tau=\min_{U}\mbox{TD}(\ketbra{GHZ}{GHZ}, U\rho U^\dagger)
\end{equation}
where TD is the trace distance and $U$ is an operator acting only on the space of the dishonest parties.

This protocol is made to be repeated several rounds until some confidence is built on the fact that the source shares GHZ states. In order to prevent the source from sending a wrong state on the round where it is supposed to be used for computation, the authors of~\cite{MEVresistant} considered randomizing this round. They also randomize which party should play the role of the Verifier at each verification round to prevent malicious actions from the parties. Thus, all parties have access to a trusted common random source that gives, at each round, a random bit $C\in\{0,1\}$ used as a security parameter and an identifier for one party $i\in[n]$. If $C=0$ (which happens with some probability $P_C$), the state is used for computation. If $C=1$ (which happens with probability $1-P_C$), the parties perform the above verification protocol with $i$ as the Verifier and restart only if the state is accepted. It has been proven that the probability that the protocol has not aborted and that a state $\rho$ such that $\mbox{TD}(\ketbra{GHZ}{GHZ}, U\rho U^\dagger)\geq \epsilon$ where $U$ is an operator on the space of the $k$ dishonest parties is used for computation is less than $P_C=\frac{4n}{k\epsilon^2}$.

The security properties in~\cite{MEVresistant} are proven in a game-based framework hence are not composable. There is for example a strategy where, when performing the protocol multiple times in a row, a malicious coalition of parties and source could increase the probability that the honest parties accept a state that is not a GHZ state. It has indeed been noticed that if we allow for a 50\% loss rate in the quantum communication, there exists a strategy for dishonest players that increases their probability of making the others accept a faulty state. This problem has been later solved in~\cite{MEVexperimental} where a loss-tolerant variation of this protocol that presents the same security properties, called the $\theta$-protocol, was implemented in a photonic setting. It mainly consists in changing the classical instructions $\Xfamily$ sent by the Verifier to angles $\Theta=\{\theta_i\}_{i=1}^n$ indicating the rotated measurement basis for each party. This protocol increases the loss that can be tolerated by the protocol, but still the dishonest parties can increase their cheating probability if the losses are high enough. For simplicity we will consider only the version of the protocol presented above (called the XY-protocol) but the following proof can be straightforwardly extended to match the $\theta$-protocol as well.\newpage

Composability issues are due to somewhat hidden assumptions in the original game-based model, such as the lossless channel assumption. Moreover, it is assumed that the dishonest parties are not disturbing the classical communication between the players or the random choosing of the Verifier and that they don't have access to the quantum memories of the honest parties. This may threaten the security of the protocol when used as a subroutine of a bigger one. Finally, the game-based framework assumes a specific strategy from the dishonest parties: they are actively trying to convince the honest parties that they all share a GHZ state when they don't. This of course makes sense when we look at the protocol in a stand-alone setting, but may not be the case when the protocol is part of a bigger, more complicated one. Using the AC framework, we can deal with all possible dishonest strategies with the help of distinguishers. It forces us to explicit every input and output of each party and to avoid hidden assumptions on dishonest behaviour and on physical resources. Additionally, it gives a box-like form for the protocol that corresponds to the way we think about applications in the quantum Internet.
\vspace{-0.2cm}
\subsection{Ideal Resource}
Let us now present the ideal resource for practical multipartite entanglement verification. Consider a source using physical resources to create and share an $n$-qubit quantum state to $n$ parties expecting a qubit from a GHZ state. Our resource, called $\MEV_C$, aims to get a sense of how trustworthy the source is, by verifying that it sends a state at least close to the GHZ state. It also has a built-in parameter $C$ that makes the resource output qubits with some probability known by all $n+1$ parties using the resource.

This black box (see Fig.~\ref{fig:IdealHonestMEV} for a 3-party example)  has $n+1$ input interfaces.  All $n$ parties wishing to test a source collectively send a start signal to the input interfaces of resource. The last interface is the source interface that gets a classical description of the state sent by the source.
Upon reception of the start inputs, $\MEV_C$ will forward the start signals to the source interface then wait for the classical description of an $n$-qubit quantum state $\rho$. After that, it outputs on all interfaces a bit $C=0$ with probability $p$, or $C=1$ with probability $1-p$. This bit indicates if the resource is going to output qubits or a verification bit $b_{out}$. The probability distribution of $C$ can be tuned freely to match any distribution. If $C=0$ it then outputs to each party a qubit of $\rho$ and if $C=1$ it computes a bit $b_{out}$ indicating if the state shared by the source is close to the GHZ state and sends it to all parties. This box is made to be composed with itself in series with a very small $p$ until all parties get a qubit or $b_{out}=1$.  \newpage

\begin{figure}[!ht]
    \centering
    \includegraphics[width=10cm]{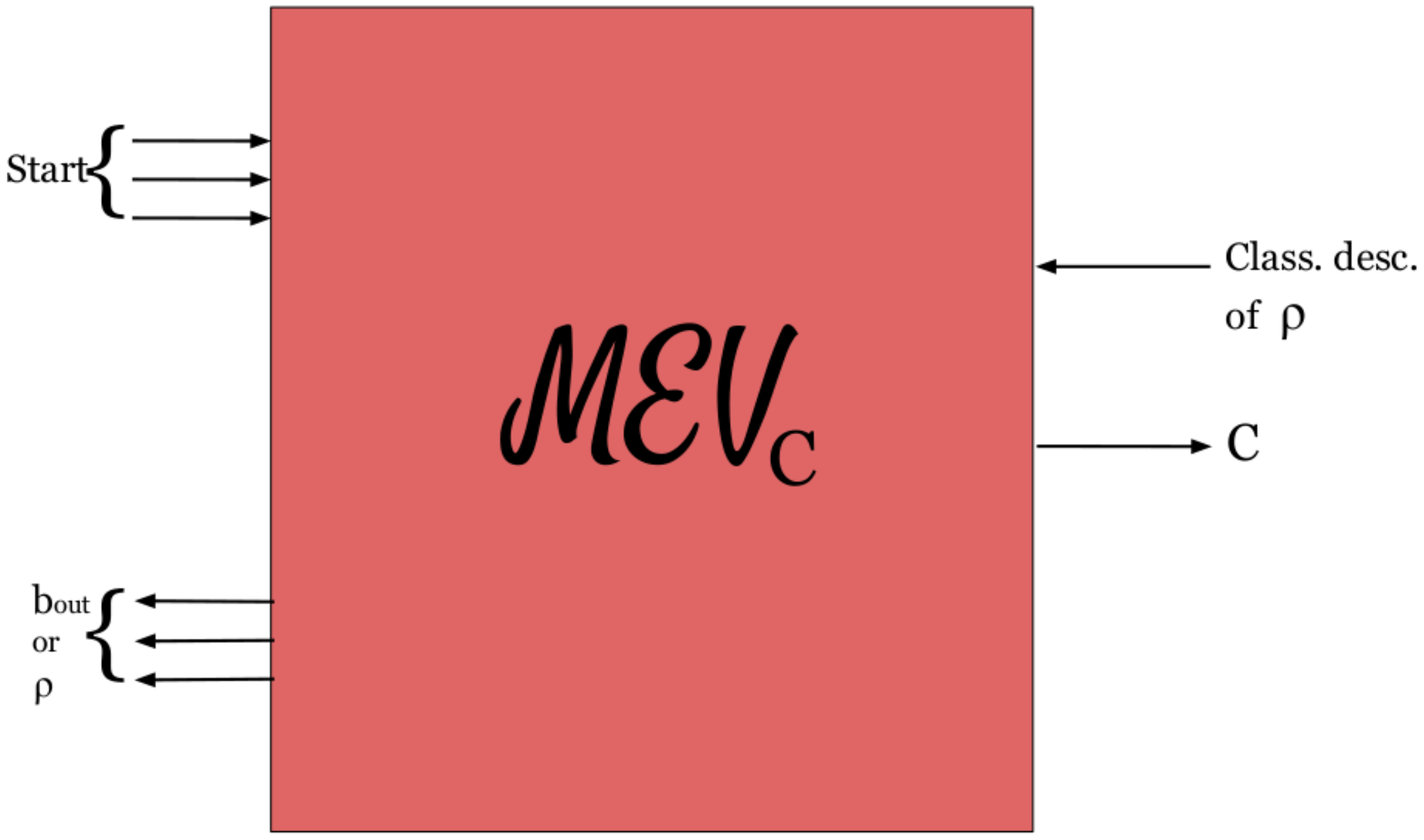}
    \caption{The $\MEV_C$ resource for $n=3$ parties. For readability we put the parties interfaces on the left and the source interface on the right. The left interfaces are ``collective interfaces'' meaning that inputs are sent collectively by all the parties and the output is obtained by all parties.}
    \label{fig:IdealHonestMEV}
\end{figure}

The output bit $b_{out}$ should indicate whether the state shared by the source is $\epsilon$-close to the GHZ state for some $\epsilon$. At this level of abstraction, we don't care whether this behaviour comes from a faulty device or an actual adversary trying to manipulate the source. Our $\MEV_C$ resource outputs a $b_{out}$ such that

\begin{equation}
    \label{output}
    b_{out} = \left\{
    \begin{array}{ll}
        0 & \mbox{with probability } 1 - \frac{\tau^{2}}{2} \\
        1 & \mbox{with probability } \frac{\tau^{2}}{2}
         
    \end{array}
    \right.
\end{equation}
with 
\begin{equation}
    \tau=\mbox{TD}(\ketbra{GHZ}{GHZ}, \rho),
\end{equation}
where TD is the trace distance. The output of the resource is thus probabilistic, and depends on the trace distance between the input state $\rho$ and the GHZ state and on the security parameter $C$. Notice that this $b_{out}$ follows the same distribution as the one of the original protocol (see Sec.~\ref{subsec:MEVprotocol}) in the case where all parties are honest. The security parameter $C$ is added to the verification procedure to make the resource suitable for practical use in larger protocols where one wants to eventually get shared entanglement between the parties when the source is acting correctly.

Now in the case of the honest use of the resource, the source interface is given as input a classical description of the GHZ state. Moreover the output $C$ remains hidden to the outside world. In AC this is modeled by using converters, the so called \textit{filters}, that block the adversarial interfaces (thus filtering the outputs) and send a specific input. In our case we define one filter $\bot$ that enforces the honest use of $\MEV_C$. It blocks any deviation from the outside world and upon reception of a start signal, it sends a classical description of a GHZ state to $\MEV_C$ (see Fig.~\ref{fig:Bot}). It has its inside interface plugged into the $\MEV_C$ resource and its outside interface open to inputs from any distinguisher (see \cite{Ratcheting} for extended discussion about filtering and the inclusion of events in AC). \newpage

\begin{figure}[!ht]
    \centering
    \includegraphics[width=8cm]{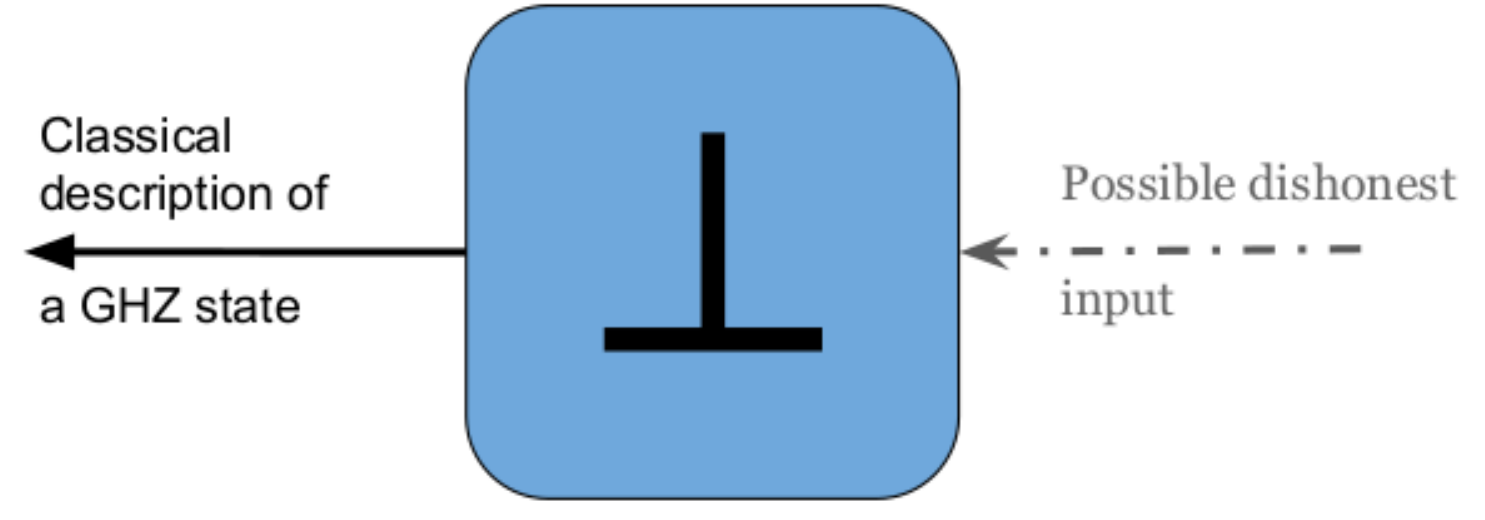}
    \caption{Filter $\bot$. Upon reception of a start input, it outputs a classical description of a GHZ state on its inside interface and blocks any input at its outside or inside interface.}
    \label{fig:Bot}
\end{figure}

Composed with $\MEV_C$, they form our ideal resource $\MEVfiltered$ for secure verified GHZ sharing or source testing (see Fig.~\ref{fig:IdealHonestFilteredMEV} for a 3-party example). 
\begin{figure}[!ht]
    \centering
    \includegraphics[width=12cm]{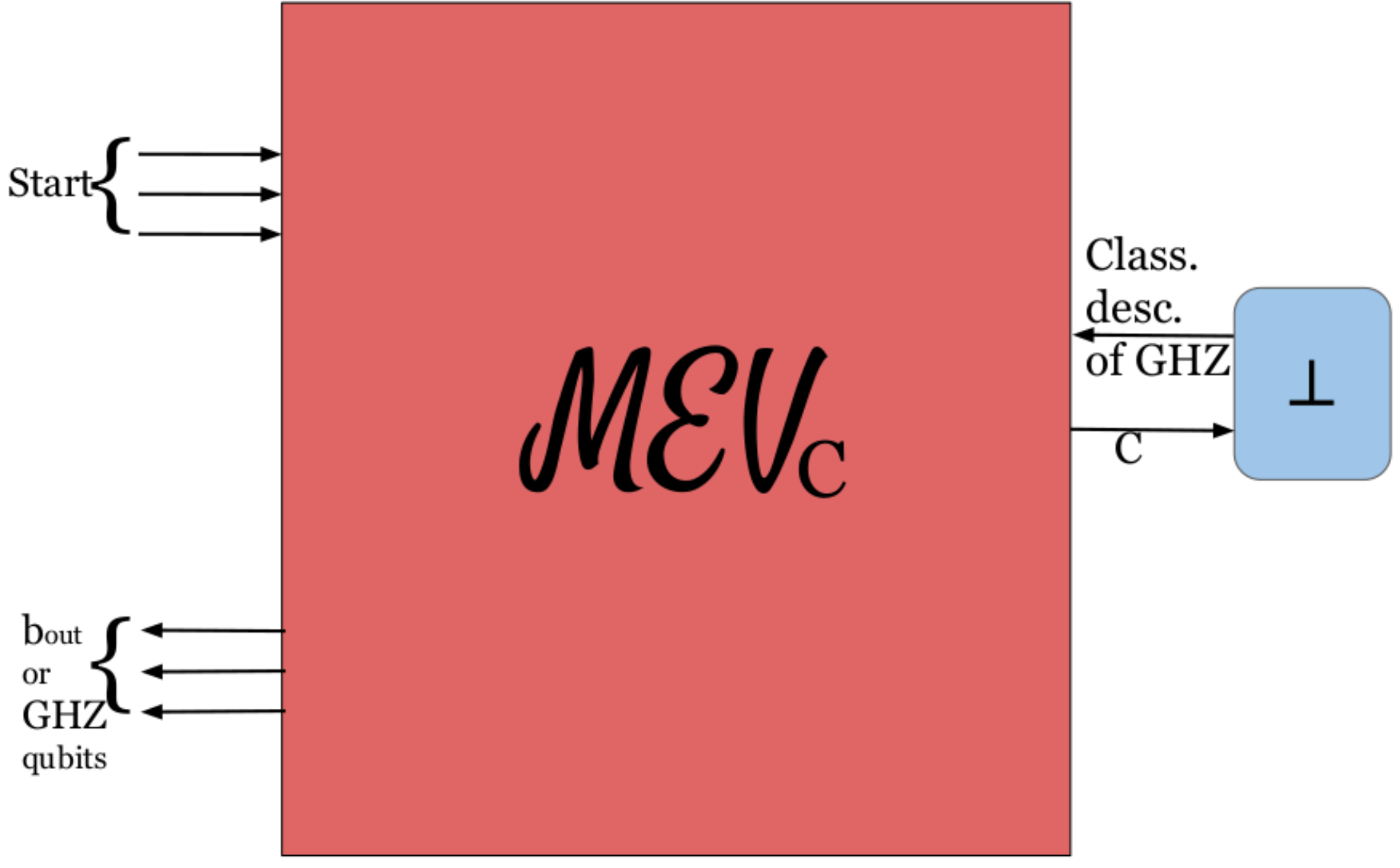}
    \caption{The ideal filtered $\MEVfiltered$ resource for $n=3$ parties. On the left are the ``collective interfaces'' that are used by the parties to collectively send the start signal and receive the output. On the right is the source interface filtered by $\bot$, that blocks any input and sends a specific message to the resource.}
    \label{fig:IdealHonestFilteredMEV}
\end{figure}

\newpage
\subsection{Concrete Resource}

We will now make explicit the protocol in the AC framework, by defining the resources used and the converters for each party. We first define the concrete resources, which in this case are abstractions of physical resources. Namely we define the state generator resource, the one-way quantum channel resource, the two-way classical channel resource and two multiparty classical computation oracles.

The state generator resource (see Fig.~\ref{fig:StateGen}) represents a perfect source of quantum states able to create arbitrary quantum states of at most $n$ qubits. Receiving a classical description of an $n$-qubit state $\rho$ on its input interface it will output each qubit of $\rho$ on its $n$ output interfaces. This resource can be used to model imperfect sources by including the noise in the classical description of the state given as input. We consider that no information is leaked by this resource about the state that it creates, as it is the more restricting scenario in our security proof. In Sec.~\ref{subsec:PracticalImp} we discuss the realization of such a resource. 

\begin{figure}[!ht]
    \centering
    \includegraphics[width=7cm]{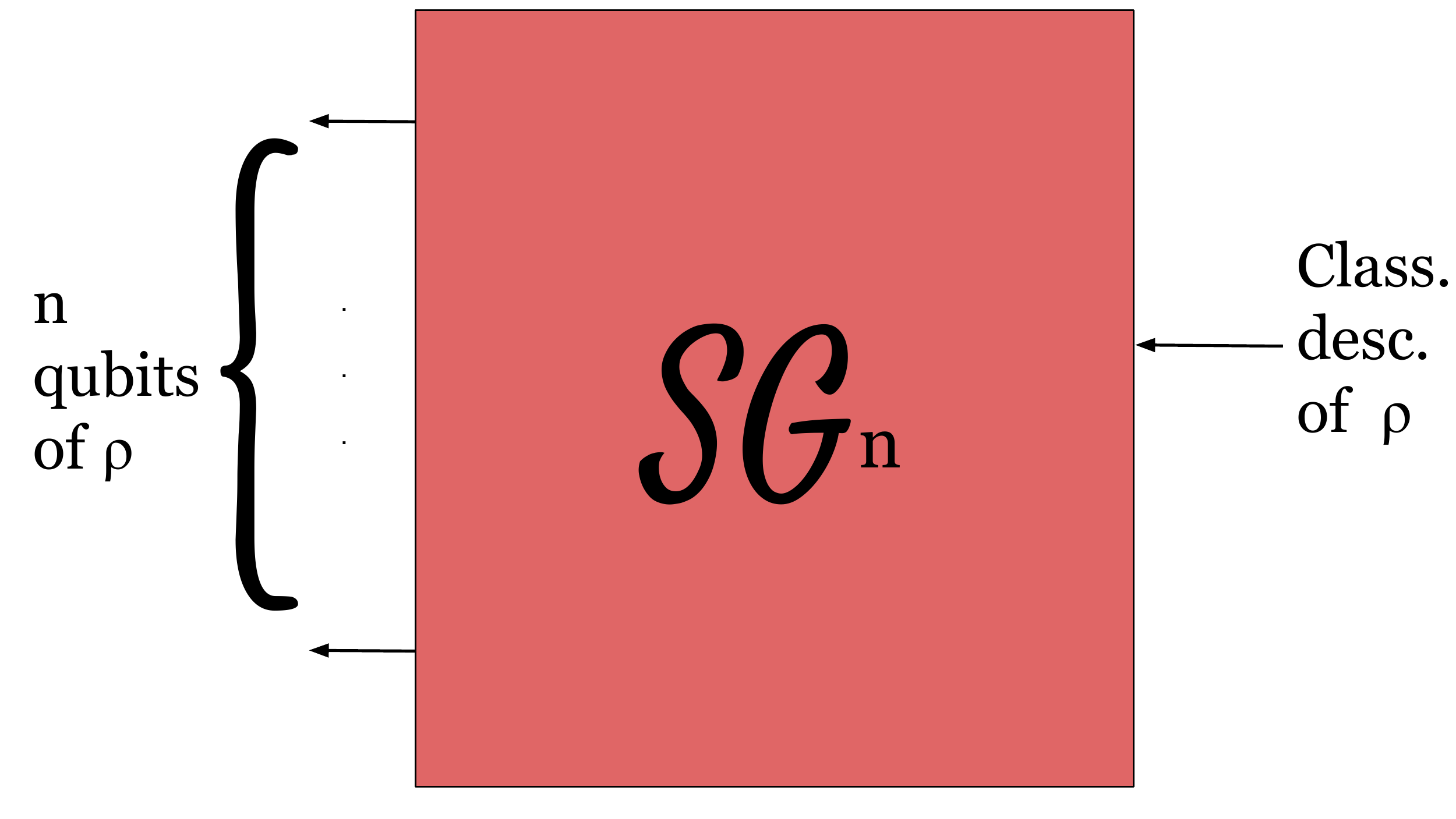}
    \caption{State generator resource.}
    \label{fig:StateGen}
\end{figure}

The $\mathcal{SG}_n$ resource is to be composed with $n$ quantum channel resources which we draw as arrows with a Q (see Fig.~\ref{fig:QChannel}). A quantum channel resource in our case is a perfect private authenticated quantum channel which takes as input a qubit and outputs the same qubit at a different place without any leakage.

\begin{figure}[!ht]
    \centering
    \includegraphics[width=7cm]{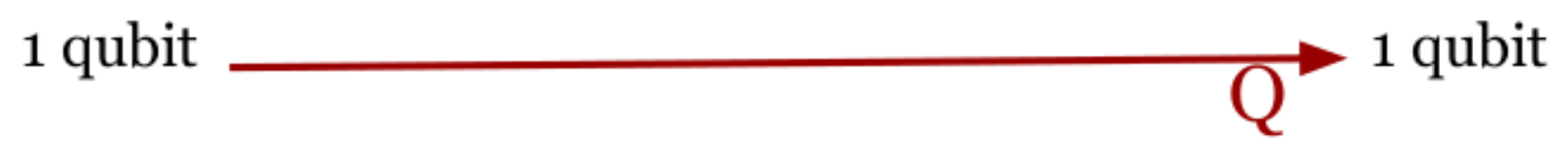}
    \caption{Quantum channel resource.}
    \label{fig:QChannel}
\end{figure}

Finally the classical communication between parties is modeled through classical channel resources which we simply draw as arrows  (see Fig.~\ref{fig:CChannel}). They take bits at any of their interfaces and transmit them to the other interface. We suppose those channels to be authenticated: to any other party watching the channel, it will also output of the message transmitted without the possibility to alter it. In order not to overload the figures, we don't represent this leaking interface when all parties are honest but we do when considering a dishonest source watching over the classical communication.

\begin{figure}[!ht]
    \centering
    \includegraphics[width=7cm]{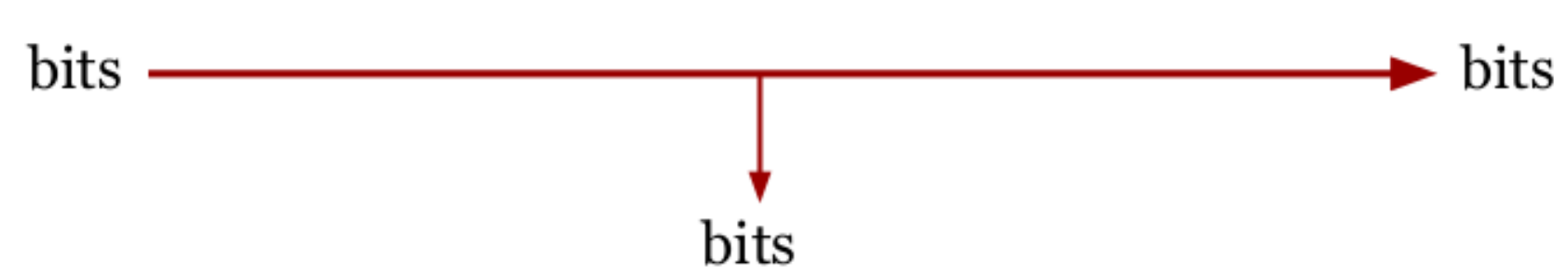}
    \caption{Classical channel resource.}
    \label{fig:CChannel}
\end{figure}
\newpage
We will abstract multiparty classical functionalities achieved by the parties by the use of oracle queries. All parties can collectively call two oracles $\mathcal{O}_C$ and $\mathcal{O}_v$ that respectively give a common random bit $C$ and a common random party identifier $v$ to each party. We will draw them as boxes with $n$ input interfaces expecting a collective query from the parties and $n$ output interface broadcasting $C$ or $v$. This is a modeling of classical communication protocols that provide random bits and random identifiers to the parties. It is not considered private and the values of $C$ and $v$ are available to any malicious party watching over the classical communication. We will discuss how these oracles can be replaced by actual classical protocols in Sec.~\ref{subsec:PracticalImp}. Moreover, each party is locally equipped with a quantum register able to perfectly store a qubit for the time required by the protocol on which they can perform one-qubit operations and measurements. Quantum registers will not be drawn in the figures for simplification purposes as well as the leakage interfaces of the classical channels, but they should not be forgotten as assumptions in our model, particularly when considering the case of a malicious party. Since we consider here all parties to be honest during the verification protocol, we only draw resources and interfaces of interest.

We call $\mathcal{R}$ the resource constructed by a state generator resource composed in series to a collection of $n$ quantum channel resources and in parallel to $n$ classical channel resources, $\mathcal{O}_C$ and $\mathcal{O}_v$. $\mathcal{R}$ formally defines the creation of a state, common classical randomness generator protocols, the (2-way) classical communication between the Verifier and the parties and the (one-way) quantum communication between the source and the parties.\\

The next step is to present the converters $\pifamily$ and $\pi_S$ that represent the protocols followed by each party and the source. They model the local computation of each party during an honest round of the protocol and can be represented either as algorithms or as boxes and arrows, that both expect some input from which they produce output to send to the resources. Their quantum abilities are equal to the ones that we give to local parties performing the multipartite entanglement verification protocol \cite{MEVresistant}.

We start with $\pifamily$ representing the protocol followed by each party $i$. $i$ is a binary identifier for each party, but for simplicity we represent it with $i\in[n]$ and we write $\pi_{[n]}$ for the parallel composition of all $\{\pi_i\}_{i=1}^n$.

\begin{protocol}{Protocol for the $i_{\text{th}}$ party $\pi_{i}$}
\begin{enumerate}
    \item Ask the source to send a GHZ state. Wait for the reception of the qubit.
    \item After reception, query $\mathcal{O}_C$, get $C$ and output it. If $C=0$, keep the qubit (output the qubit to party).
    \item If $C=1$, \begin{enumerate}
        \item Query $\mathcal{O}_v$, get $v$.
        \item If $v\neq i$ \begin{enumerate}
            \item Wait for the reception of $x_i$.
            \item If $x_{i}=0$, perform a Hadamard operation on the qubit. If $x_{i}=1$, perform a $\sqrt{X}$ operation on the qubit. 
            \item Measure in the $\{\ket{0},\ket{1}\}$ basis.
            \item Send the outcome $y_{i}$ to the Verifier via the classical channel resource.
            \end{enumerate}
        \item If $v=i$ \begin{enumerate}
            \item Create a random bit string $X=\{x_{i}\}_{i=1}^{n}$ with $x_{i}\in\{0,1\}$ such that $\sum_{i=1}^{n}x_{i}\equiv 0$ (mod 2) 
            \item Send $x_{i}$ it to party $i$ via a classical channel resource, keep $x_{v}$.
            \item Follow steps (iii).b.2 to (iii).b.4 and get $y_v$
            \item Wait for the reception of all the other $y_{i}$.
            \item  Upon the reception of all the $y_{i}$, output 0 to all if 
            \begin{equation*}
    \sum_{i=1}^{n}y_{i}\equiv\frac{1}{2}\sum_{i=1}^{n}x_{i} \text{(mod 2)}
            \end{equation*}
            and 1 otherwise.
          \end{enumerate}
\end{enumerate}
\end{enumerate}

\end{protocol}

The actual verification protocol is thus seen here as a subroutine (steps (iii).(a) to (iii).(c)). All parties start by collectively querying a qubit and $C$ and then, depending on the value of $C$, they either keep the qubit or do the verification protocol. During the verification protocol, one party is chosen to be the Verifier and after some classical communication and local quantum operations, the Verifier sends the output $b_{out}$ to all parties.

The last converter, $\pi_S$, represents the local operation that an honest source would perform using the source to create an $n$-qubit GHZ state and send it to the parties. That is simply sending, upon receiving a signal from the parties, a classical description of the GHZ state to the $\mathcal{SG}_n$ resource. It implies that the source is not watching the classical communication between the parties at any point. Functioning like a filter, this converter is made to be removed in case the source is noisy or some malicious party takes control of the source to reveal new interfaces to the outside world.

\begin{protocol}{Protocol for the source $\pi_{S}$ }
\begin{enumerate}
    \item  Upon reception of a query by the parties, send a classical description of the GHZ state to the $\mathcal{SG}_n$ resource.
\end{enumerate}
\end{protocol}
\vspace{1cm}
Together with $\R$, this completes the definition of the concrete multipartite entanglement resource $\ConcreteProt$ (see Fig.~\ref{fig:ConcreteProt} for a 3-party example), which takes as input a start signal and outputs a bit $C$ then a qubit from a GHZ state to each party or a bit $b_{out}=0$.

\begin{figure}[!ht]
    \centering
    \includegraphics[width=16.5cm]{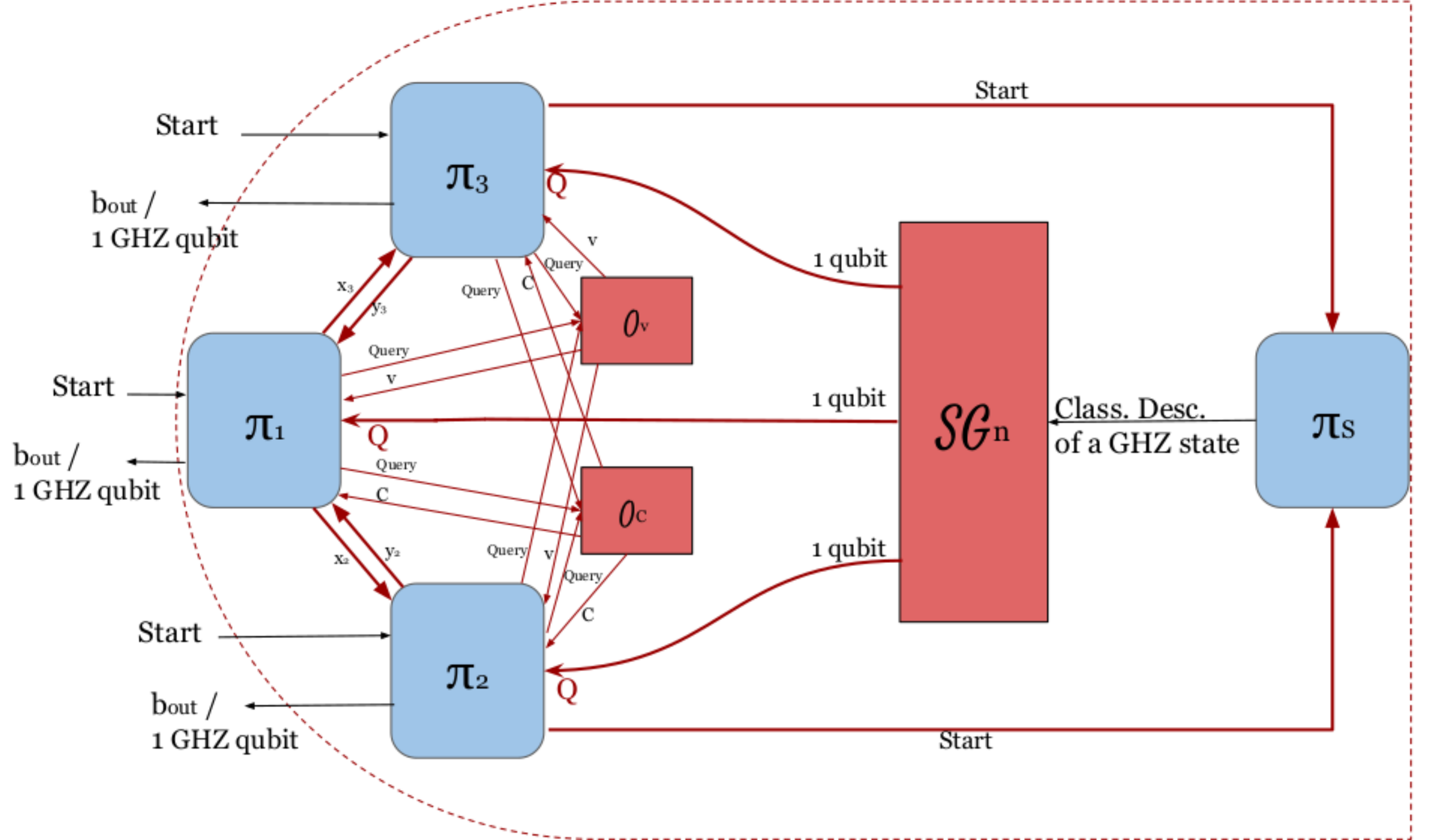}
    \caption{The $\ConcreteProt$ Resource within the dotted red line for $n=3$ parties wishing to test a source, when party 1 is chosen to be the Verifier. We represent resources in red and converters in blue. We recall the timeline of the protocol: (1) all the $\pi_i$ send a start signal to $\pi_S$ that sends a classical description of a GHZ state to the $\mathcal{SG}_n$ resource. (2) Upon reception of the qubit, they send a query to $\mathcal{O}_C$ and get $C$. (3) If $C=0$ they output a GHZ qubit and if $C=1$ the parties query $\mathcal{O}_v$ and get $v$ (here party 1). (4) The Verifier sends instructions $\Xfamily$ (here $\{x_2,x_3\}$) to others parties, get outcomes $Y=\{y_i\}_{i=1}^n $ (here $\{y_2,y_3\}$) and computes and broadcasts $b_{out}$. To avoid overloading the figure we don't represent quantum memories as well as the classical signals going from $\pi_1$ to $\pi_S$. As $\pi_S$ represents honest behaviour from the source, we also don't represent the leakage of information from the classical channels.}
    \label{fig:ConcreteProt}
\end{figure}\newpage

\subsection{Security Analysis}
\label{subsec:SecAnal}
We come now to the proof of the main claim of this paper, namely that the multipartite entanglement verification protocol $\pi$ securely constructs the $\MEV_C$ resource out of $\R$. We proceed as expected from the security definition of Sec.~\ref{subsec:SecDef} that is by finding simulators to emulate local dishonest behaviour on the ideal resource. A dishonest behavior from a party is simply modeled by removing the associated converter and making new free interfaces accessible to a distinguisher. Simulators should render the ideal resource indistinguishable from dishonest concrete resources.

We will only consider cases that are of interest for our security claim which are when all parties are honest and when the source is noisy or malicious. The case of dishonest parties possibly tampering the source is discussed in Sec.~\ref{subsec:MaliciousParty}, but it appears that composable security cannot be proven in the AC framework when a party is dishonest. Distinguishers in this section are all powerful, both classicaly and quantumly.

\subsubsection{Correctness.}
The first step of the proof corresponds to the correctness of the multipartite entanglement verification protocol, meaning that when all parties are honest and the source is honest the parties all get either a qubit from a GHZ state or a bit $b_{out}=0$.

\begin{theorem}
The multipartite entanglement verification protocol emulates the filtered ideal resource $\MEVfiltered$.
\end{theorem}

\begin{proof}
Let $\mathbf{D}$ be an all powerful distinguisher trying to guess between $\MEVfiltered$ and $\ConcreteProt$. Let us look at the distribution of outputs that it will get from them.

$\mathbf{D}$ first sends start signals to both resources. When interacting with $\MEVfiltered$, it gets $C=1$ and $b_{out}=0$ with some probability $1-p$ and  $C=0$ and $n$ qubits from a GHZ state with probability $p$ by definition of our resource. Throughout this paper, the probability distribution of $p$ is tuned to match the one of $\mathcal{O}_C$. When interacting with $\ConcreteProt$, the distinguisher thus performs the concrete multipartite entanglement verification protocol with the same probability $p$. If all parties share a GHZ, the condition $\sum_{i=1}^{n}y_{i}\equiv\frac{1}{2}\sum_{i=1}^{n}x_{i} \text{(mod 2)}$ is always fulfilled (see~\cite{MEVresistant} for complete proof) so the Verifier always sends $b_{out}=0$ at the end. Hence, $\mathbf{D}$ gets $C=1$ and $b_{out}=0$ with probability $1-p$ and $C=0$ and $n$ qubits from a GHZ state with probability $p$.

We can conclude that for any distinguisher $\mathbf{D}$,  $d^{\mathbf{D}}(\ConcreteProt,\MEVfiltered)=0$ hence
\begin{equation}
    \ConcreteProt\approx\MEVfiltered.
\end{equation}
\end{proof}
\newpage
\subsubsection{Dishonest source.}
Let us now look at the case of a dishonest or noisy source. As custom in AC, we model this by removing the filter $\bot$ of the ideal resource and the protocol $\pi_S$ of the concrete one (see Fig.~\ref{fig:ConcreteProtDis}). This leaves a new interface free for a distinguisher to send in a classical description of a state $\rho$. Because we do not use private but rather authenticated classical communication, the distinguisher also receives all leakage of classical communication between the parties and when they query oracles.

\begin{figure}[!ht]
    \centering
    \includegraphics[width=16cm]{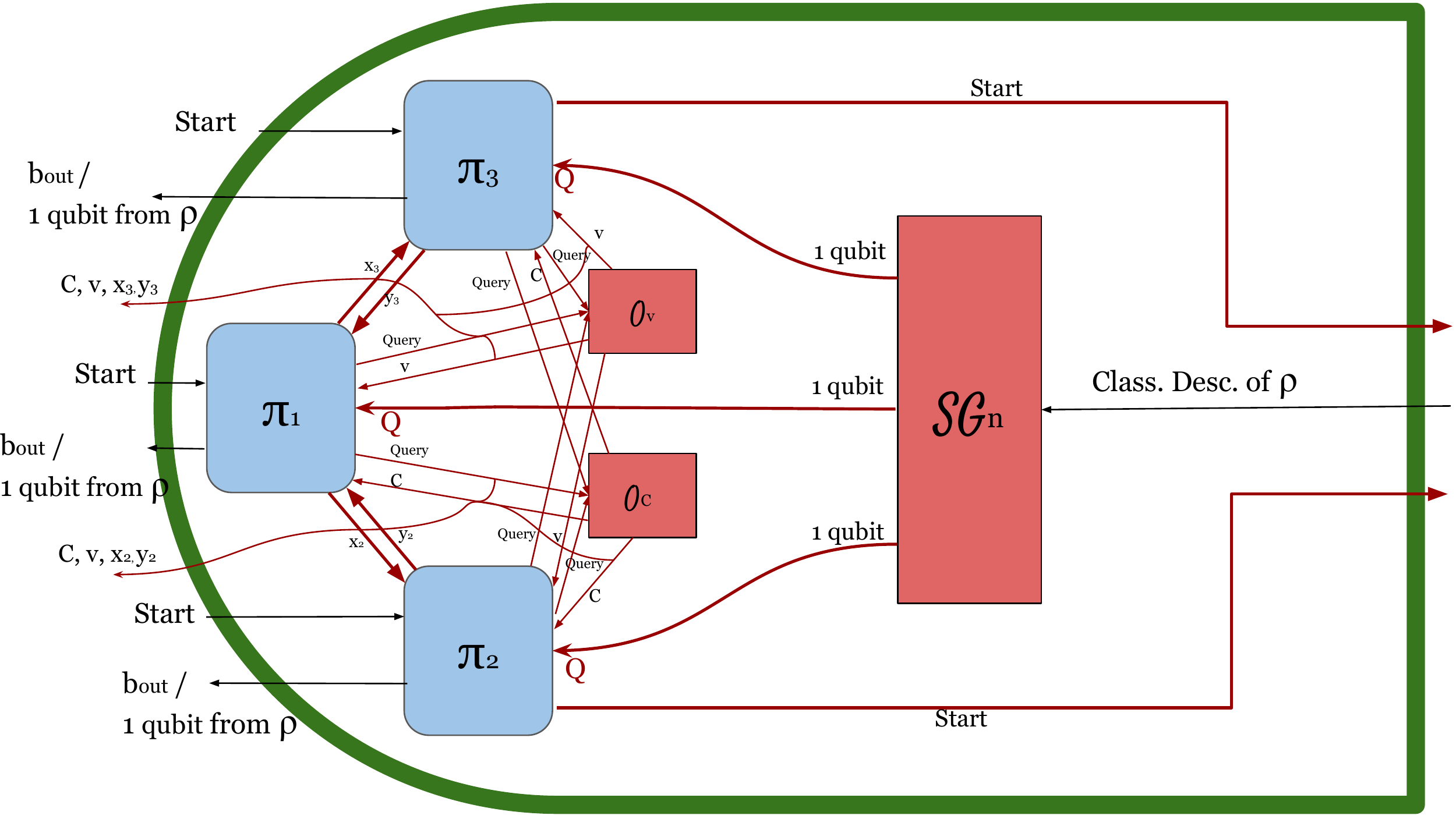}
    \caption{The $\pi_{[n]}\mathcal{R}$ resource for $n=3$ parties when party 1 is chosen as the Verifier, accessed by a distinguisher (in green). To not overload the figure we join all leakages interfaces from the classical channel resources into two arrows, but they should each be considered as a different interface the distinguisher has access to.}
    \label{fig:ConcreteProtDis}
\end{figure}

In order to prove security, as expected from the security definition of Sec.~\ref{subsec:SecDef}, we need to find a simulator $\sigma_{S}$ such that we can prove $\pi_{[n]}\mathcal{R}\approx\MEV_C\sigma_S$. It should emulate dishonest behaviour and the new interfaces a distinguisher has access to when interacting with the ideal resource. \\

Let $\sigma_{S}$ be the simulator shown in Fig.~\ref{fig:SigmaS}.
It first takes as input a start signal from the $\MEV_C$ resource, then emulates the verification protocol by forwarding this start signal. After receiving a classical description of a state $\rho$, it forwards it to $\MEV_C$ and gets and forwards the bit $C$. If $C=1$, it creates a random $v \in [n]$ and a random bit string $\Xfamily$ such that $\sum_{i=1}^{n}x_{i}\equiv 0$ (mod 2) and sends them to the outside world, except for $x_v$. Then it computes a table of possible measurement outcomes by calculating all necessary scalar products:
\begin{multline}\\
\label{eq:outputtable}
    Pr[y_1=0,y_2=0,...,y_n=0]=\bra{00...0} U \rho U^{\dagger}\ket{00...0}\\
    Pr[y_1=0,y_2=0,...,y_n=1]=\bra{00...1} U \rho U^{\dagger}\ket{00...1} \\
    ...\\
    Pr[y_1=1,y_2=1,...,y_n=1]=\bra{11...1} U \rho U^{\dagger}\ket{11...1}\\
\end{multline}
with $U=H^{x_1}(\sqrt{X})^{1-x_1}\otimes H^{x_2}(\sqrt{X})^{1-x_2}\otimes ... \otimes H^{x_n}(\sqrt{X})^{1-x_n}$ corresponding to the local operations made by each party on their qubit in the verification protocol. Then it randomly samples $Y=\{y_i\}_{i=1}^n$ from this table and sends them to the outside world, except for $y_v$.\\
\begin{figure}[!ht]
    \centering
    \includegraphics[width=8cm]{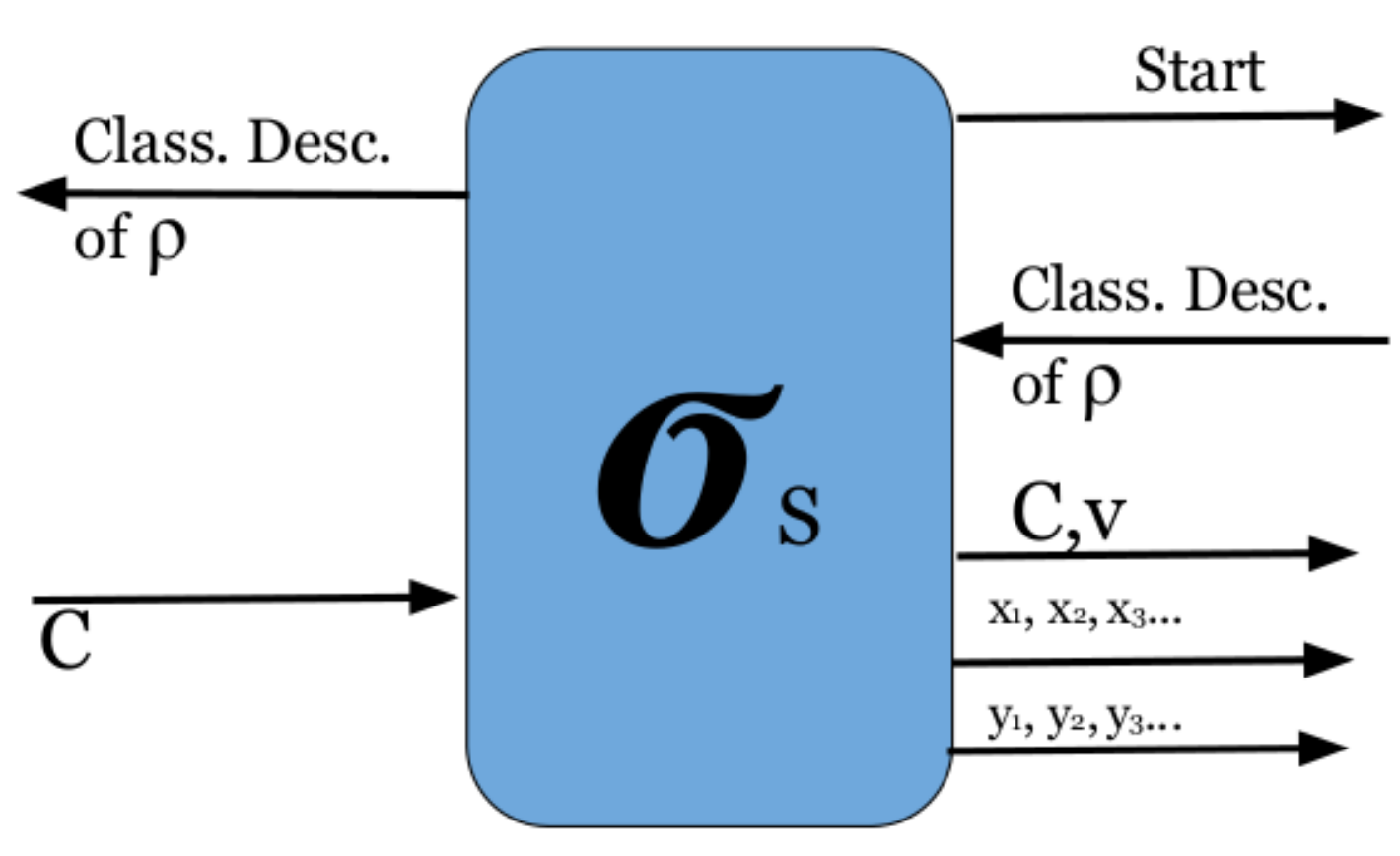}
    \caption{Simulator $\sigma_S$ for a dishonest source.}
    \label{fig:SigmaS}
\end{figure}\\
Roughly speaking, $\sigma_S$ classically emulates the whole multiparty protocol by reproducing the classical communication and local quantum operations. Plugged in $\MEV_C$, this defines a new resource $\MEV_C\sigma_S$ (see Fig.~\ref{fig:IdealProtDis}). With this simulator we can state that:

\begin{theorem}
The multipartite entanglement verification protocol with a noisy or malicious source emulates the ideal resource $\mathcal{MEV}_C\sigma_{S}$.
\end{theorem}

\begin{proof}
In this scenario, we have to prove an equivalence between $\MEV_C\sigma_S$ and $\pi_{[n]}\mathcal{R}$ by showing that no distinguisher sending inputs and receiving outputs from both can guess which resource it is interacting with. In the concrete setting this means that the parties will share a state $\rho$ that is $\tau$-close to the GHZ state, with $\tau= \mbox{TD}(\ketbra{GHZ}{GHZ},\rho)$, that they will either keep or verify with probability S. In~\cite{MEVresistant}, it is shown that a state $\rho$ passes the verification test with probability $1-\frac{\tau^{2}}{2}$. \newpage
\begin{figure}[!ht]
    \centering
    \includegraphics[width=15cm]{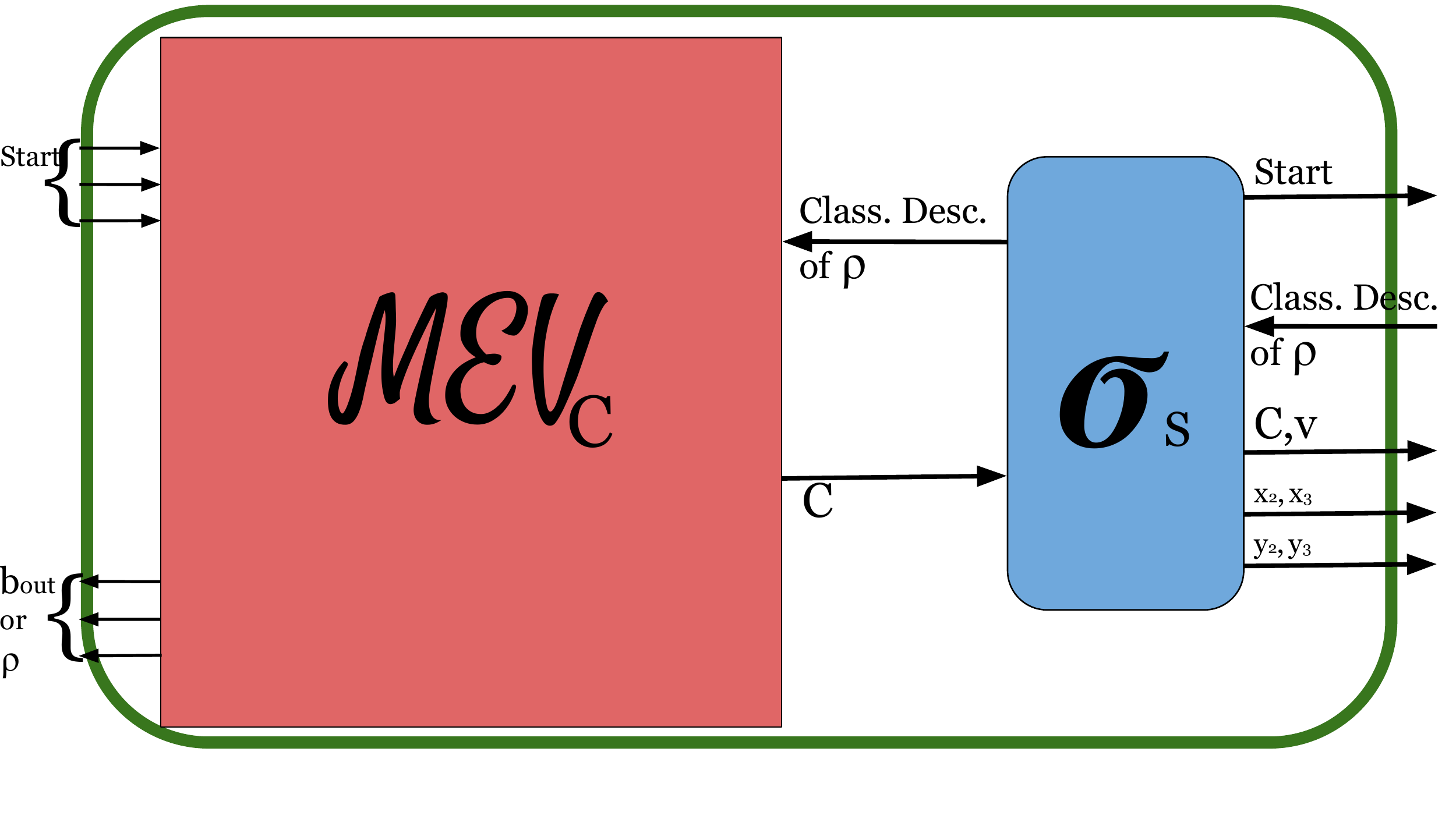}
    \caption{The $\MEV_C\sigma_S$ resource for $n=3$ parties accessed by a distinguisher (in green).}
    \label{fig:IdealProtDis}
\end{figure}
Let $\mathbf{D}$ be an all powerful distinguisher trying to guess between $\pi_{[n]}\R$ and $\MEV_C\sigma_{S}$. In the concrete setting, it sends in start signals at the parties interfaces then receives it at the source interface and sends a classical description of a state $\rho$ to $\mathcal{SG}_n$. $\mathbf{D}$ then sees all the classical communication happening out of the authenticated classical channels. More explicitly it will first see a bit $C$ and if $C=0$, nothing but the qubits of $\rho$ at each parties interface. If $C=1$, a random identifier $v\in[n]$ leaks, then random bits $X\backslash\{x_v\}$ from the Verifier to each party, then the outcome of each party's measurement except the Verifier's $Y\backslash\{y_v\}$ and finally the bit $b_{out}$ broadcasted by the Verifier.

In the ideal scenario, after $\mathbf{D}$ sends in a start signal, $\MEV_C$ forwards it to $\sigma_{S}$ which then sends a start signal simulating the query of a state by the parties. After that, the distinguisher sends a classical description of a state $\rho$ to $\sigma_S$ who forwards it to $\MEV_C$, which outputs $C$ at all its  interfaces. $\sigma_S$ gets $C$ and outputs it at its outside interface. If $C=0$, $\MEV_C$  outputs the qubits of $\rho$ at each party's interface. If $C=1$, $\sigma_S$ creates and outputs a random $\hat{v}\in[n]$ then computes a random bit string $\hat{X}=\{\hat{x}_i\}_{i=1}^n$ such that $\sum_{i=1}^{n}\hat{x}_{i}\equiv 0$ (mod 2) and sends them to the outside world, except for $\hat{x}_v$. After that, $\sigma_S$ computes the table of Eq.~(\ref{eq:outputtable}), randomly samples $\hat{Y}=\{\hat{y}_i\}_{i=1}^n$ and outputs them all to the outside world except for $\hat{y}_v$. Finally $\MEV_C$ outputs $\hat{b}_{out}=0$ with probability $1-\frac{\tau^{2}}{2}$ and $\hat{b}_{out}=1$ otherwise.

The probability distribution of the bit $C$ is designed to match the probability distribution given by the oracle $\mathcal{O}_C$. In the concrete setting $v$ is chosen randomly among the players through a query to the oracle $\mathcal{O}_v$ so we have that for all $i\in [n]$, $\Pr[v=i]=\Pr[\hat{v}=i]$. $\Xfamily$ and $\hat{X}=\{\hat{x}_i\}_{i=1}^n$ are both chosen randomly so their probability distribution is the same. $Y=\{y_i\}_{i=1, i \neq v}^n$ are the outcomes of the measurements of each qubit $\rho$ by each party in the $\{\ket{0},\ket{1}\}$ basis after doing the operation indicated by each $x_i$. The state after each party applied their operation is $U \rho U^{\dagger}$ with $U=H^{x_1}(\sqrt{X})^{1-x_1}\otimes H^{x_2}(\sqrt{X})^{1-x_2}\otimes ... \otimes H^{x_n}(\sqrt{X})^{1-x_n}$. They are samples from the table of Eq.~(\ref{eq:outputtable}). Hence for each $i\in[n]$ we have that $\Pr[y_i=0]=\Pr[\hat{y}_i=0]$. Finally, by definition of our $\MEV_C$ resource, the probability distribution of $\hat{b}_{out}$ is the same as the one of $b_{out}$.

The probability distribution of the output given by the two resources depending on the inputs is thus the same. Hence we have that for any distinguisher $\mathbf{D}$, $d^{\mathbf{D}}(\pi_{[n]}\mathcal{R}, \mathcal{MEV}\sigma_{S})=0$ and 
\begin{equation}
    \pi_{[n]}\mathcal{R}\approx\mathcal{MEV}_C\sigma_{S}.
\end{equation}

\end{proof}

\subsubsection{Conclusion.}
We have proved that $\ConcreteProt\approx\MEV_C\bot$ and that $\exists \sigma_S$ s.t. $\pi_{[n]}\mathcal{R}\approx\MEV_C\sigma_S$. This means that the multipartite entanglement verification protocol presented is composable when all parties are honest but with a possibly dishonest source. The protocol can thus be thought of as a black box and equivalently replaced by the $\MEV_C$ resource (Fig.~\ref{fig:IdealHonestMEV}) when designing protocols using this one as a subroutine. It assumes that the parties have access to resources $\mathcal{R}$, including common oracles and quantum memories. 

\subsubsection{Application : Verified GHZ sharing resource.} The composability result we proved allows $n$ parties to securely compose the protocol with itself multiple times. If the probability that $C=0$ is sufficiently small, the protocol will be repeated on expectation enough rounds to allow the parties to build high confidence on the source's ability to create a state close to the GHZ state. Since the round where they will actually use the qubits sent by the source to perform some communication or computation protocol is unknown to the source, it is not possible for the source to adapt and decide when to send faulty states. Hence it forces the source to send states that are sufficiently close to the GHZ state every time it is queried. We call this protocol the multi-round multipartite entanglement verification protocol.\\

By defining converters $\{\Pi_i\}_{i=1}^n$ representing the aforementioned protocol, we can construct a resource $\Pi_{[n]}\MEVfiltered$ that gives either a state at least $\epsilon$-close to the GHZ state to $n$ parties or an abort signal (see Fig.~\ref{fig:MultiRound} for a 3-party example and the explicit description of a $\Pi_i$). As it is a composable framework, AC allows us to state that
\begin{align}
    \Pi_{[n]}\ConcreteProt\approx\Pi_{[n]}\MEV_C\bot \\
    \textnormal{and }
    \exists \sigma_S \textnormal{ s.t. } \Pi_{[n]} \pi_{[n]} \mathcal{R} \approx \Pi_{[n]}\MEV_C \sigma_S.
\end{align} \newpage
\begin{figure}[!ht]
    \centering
    \includegraphics[width=18cm]{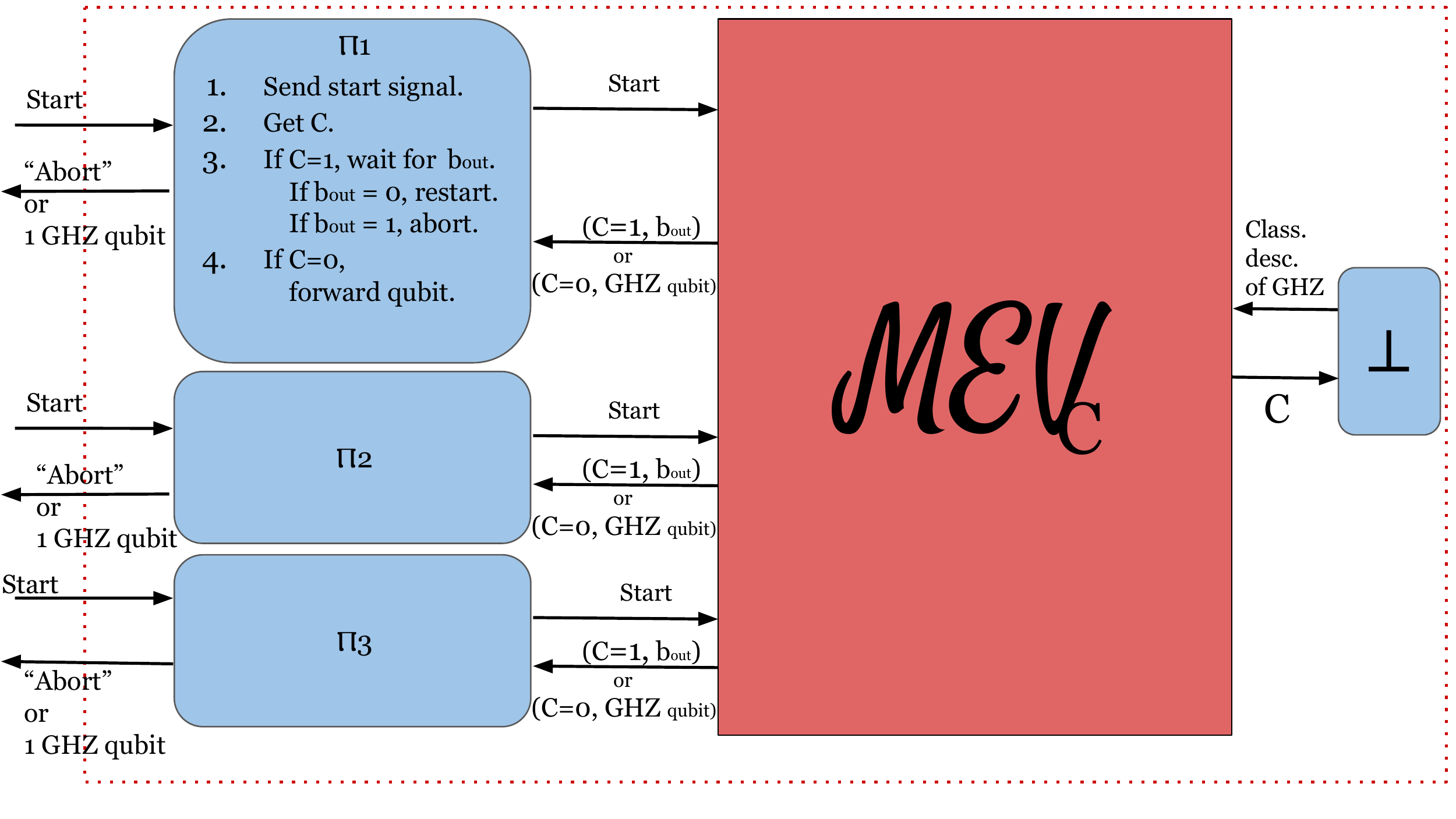}
    \caption{Multi-round verification resource $\Pi_{[n]}\MEVfiltered$ for 3 parties (in the red dotted square). It takes start signals as input and outputs either a shared quantum state $\epsilon$-close to the GHZ state or an abort signal. }
    \label{fig:MultiRound}
\end{figure}
Let us define a {\em verified GHZ state sharing} resource that we call $\mathcal{GHZ}$ (see Fig.~\ref{fig:GHZ}). This resource is the idealisation of multipartite entanglement verification achieved through an interactive protocol between the source and the parties. We assume that at each round of the interaction a state is produced and shared by the source and the parties perform some verification protocol until, in the end, they decide to trust that the shared state is close to the GHZ state or abort the protocol. 
$\mathcal{GHZ}$ takes as input start signals from the parties, then interacts with the source and finally outputs either a state $\epsilon$-close to the GHZ state or an abort signal. The interaction is abstractly modeled in the following way: first a Start signal is sent to the source interface, which replies with the classical description of a state $\rho$. Then $\mathcal{GHZ}$ will either ask for another state by sending a ``Continue'' signal to the source interface, or output an ``Abort'' signal to all interfaces because the current state was found to be far from the GHZ state, or, last, stop the protocol and share the last state it has received to the parties interface and send a ``Stop'' signal to the source interface.

This resource abstracts all the local operations and communication between the parties. From their point of view it is simply a source of states that are close to the GHZ state. However, to capture possibly malicious behavior from the source, we include the interaction on the source interface. We argue this is an abstract enough resource that captures all interactive verification procedures where the parties verify a number of states from the source before asserting that the source gives close to GHZ states.

\begin{figure}[!ht]
    \centering
    \includegraphics[width=14cm]{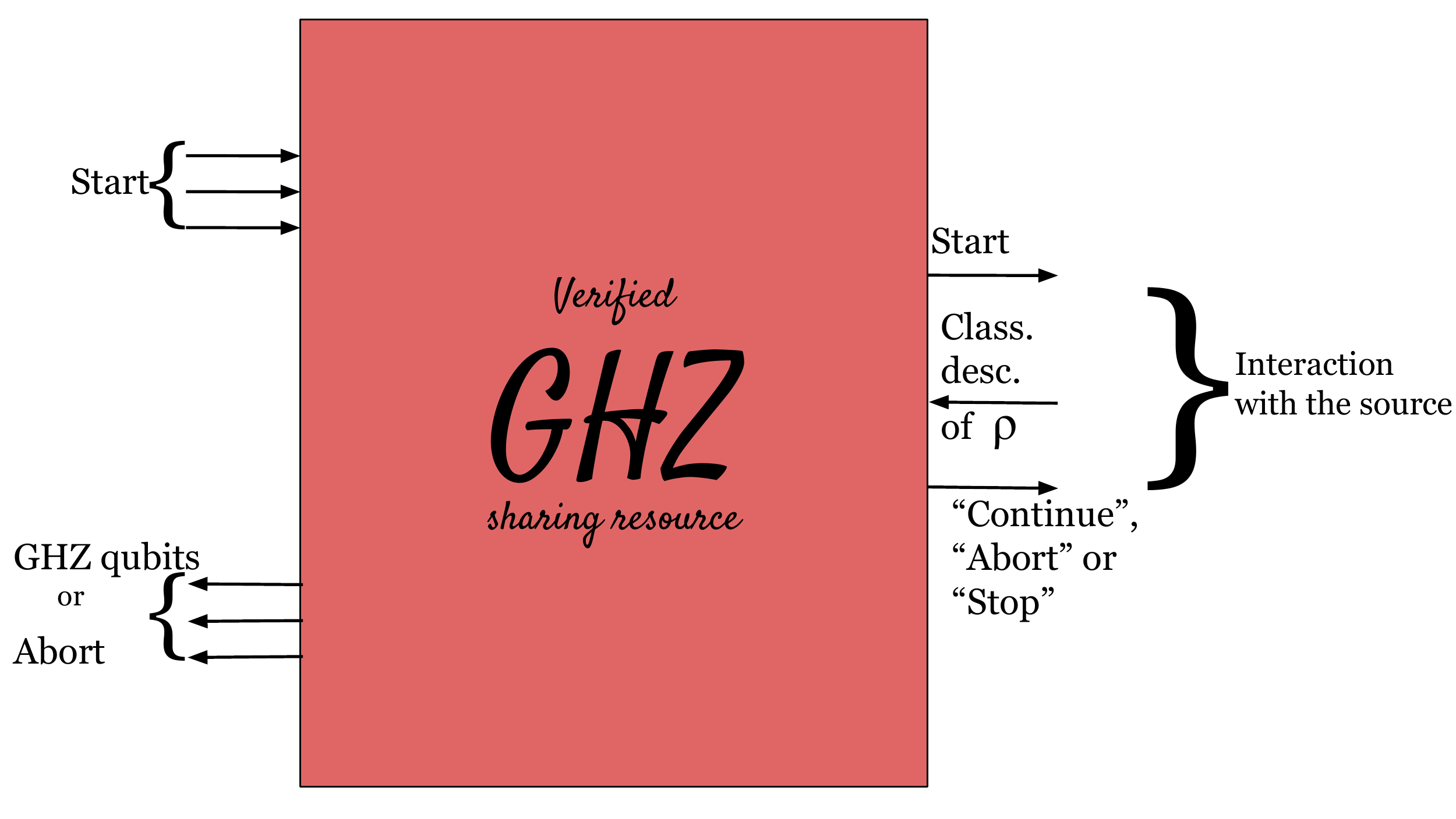}
    \caption{Verified GHZ sharing resource for 3 parties. It takes start signals as input from the parties on the left interface then interacts with the source on the right interface. It outputs either a shared quantum state $\epsilon$-close to the GHZ state or an abort signal to the parties.} 
    \label{fig:GHZ}
\end{figure}

Similarly to Sec.~\ref{subsec:SecAnal}, when we define $\bot'$ and $\sigma_C$ as in Fig.~\ref{fig:GHZproof}, we can prove that
\begin{align}
    \Pi_{[n]}\MEV_C\bot \approx \mathcal{GHZ}\bot'\\
    \textnormal{and } \Pi_{[n]}\MEV_C \approx \mathcal{GHZ}\sigma_C
\end{align}

Hence,
\begin{align}
    \Pi_{[n]}\ConcreteProt\approx \mathcal{GHZ}\bot\\
    \textnormal{and }
    \exists \sigma_S \textnormal{ s.t. } \Pi_{[n]} \pi_{[n]} \mathcal{R} \approx\mathcal{GHZ}\sigma_S
\end{align}
\vspace{1cm}

This means that the multi-round multipartite entanglement verification protocol constructs the $\mathcal{GHZ}$ resource out of $\mathcal{R}$. We can also state that it is composably secure in the setting of all honest parties and in the presence of a possibly malicious or noisy source. We can conclude that this protocol allows $n$ parties to get a GHZ state as a subroutine of a bigger protocol with an untrusted source.

\begin{figure}[!ht]
\centering
\begin{minipage}{0.45\textwidth}
    \includegraphics[width=0.8\textwidth]{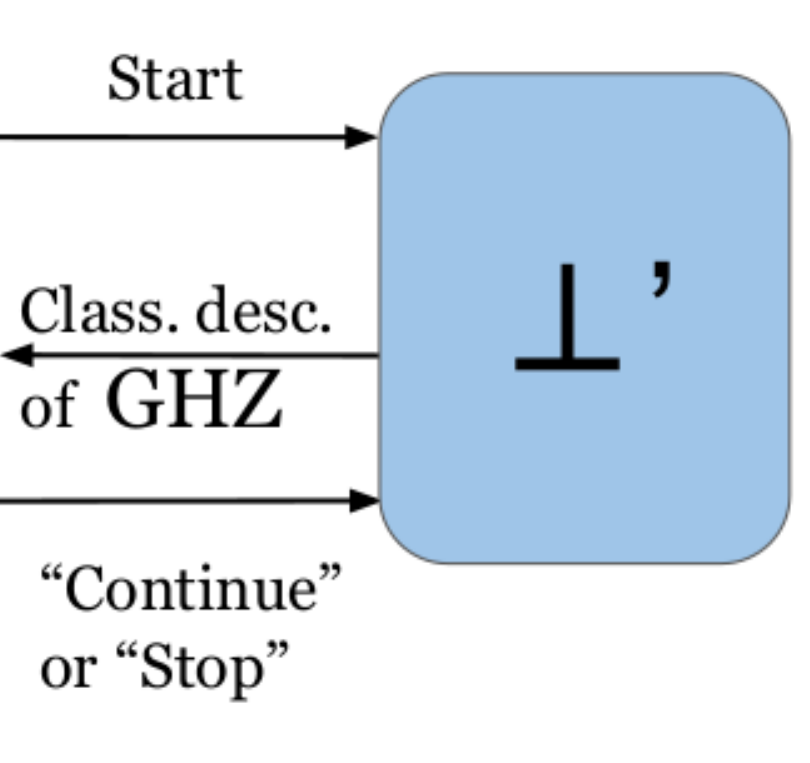}
\end{minipage}
\begin{minipage}{.45\textwidth}
    \includegraphics[width=\textwidth]{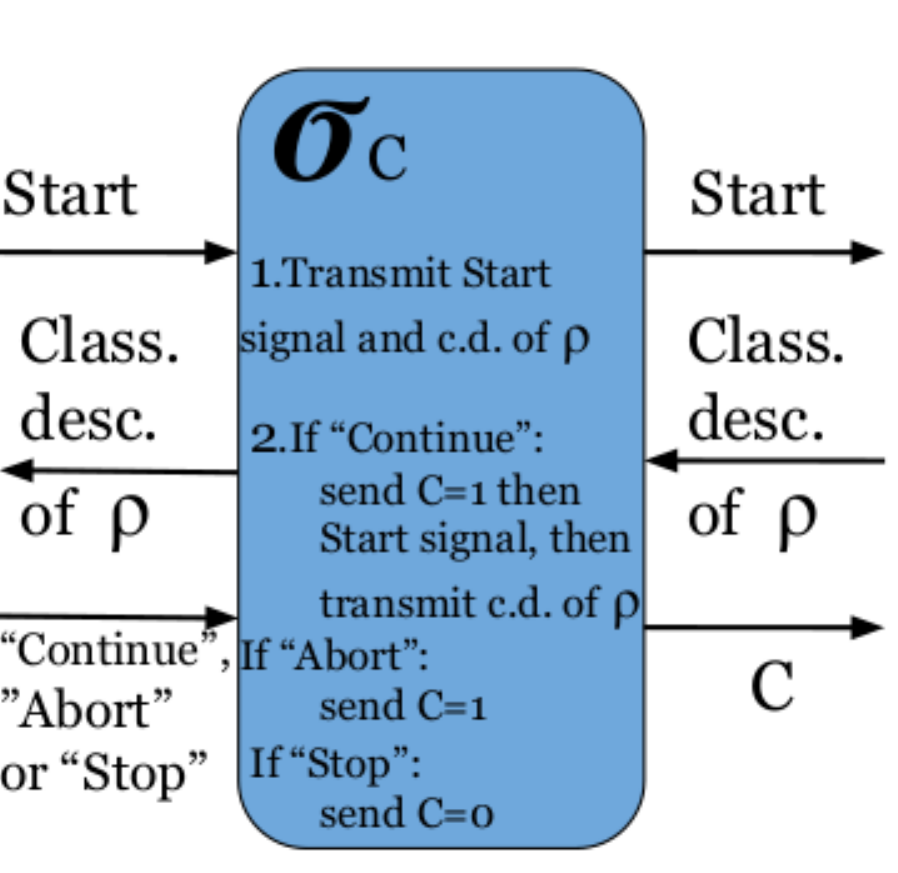}
\end{minipage}
    \caption{Filter $\bot'$ (on the left) and simulator $\sigma_C$ (on the right) to plug into the $\mathcal{GHZ}$ resource. The former represents the honest use of the resource and allows us to state the correctness of the AC security proof (Eq. (12)). The latter is the simulator that models the interface to which a distinguisher has access when we consider the source to act maliciously using the multiround multipartite entanglement verification resource $\Pi_{[n]}\MEV_C$. It allows us to derive the second part (noisy or malicious source) of the AC security proof (Eq. (13)).}
    \label{fig:GHZproof}
\end{figure}
\newpage

\section{Discussion}
\vspace{-3mm}
\subsection{Case of honest parties}
\vspace{-1.5mm}
The multipartite entanglement verification protocol is particularly suited in a distributed computing scenario where the parties are honest but where there could be a faulty resource. They can use this protocol to check if the noise of an entanglement source is small enough for practical use. Indeed, if after many rounds of performing this protocol the output is most of the time $b_{out}=0$, they can realistically be sure that the source is producing states that are close to the GHZ state. Its composability allows for the construction of the multi-round verification resource, which can find practical use in larger communication protocols, as for example in anonymous ranking~\cite{Ranking}, quantum secret sharing~\cite{SecretSharing} or distributed consensus~\cite{Consensus} protocols. In fact, any protocol that starts with a GHZ state shared among $n$ honest parties that don't trust their source can be composed with this one in a secure way. This might seem limiting but is in fact realistic in many distributed computing settings. This protocol can also be seen as a building block of a quantum network. We can reasonably assume that parties are honest when performing protocols establishing the network in the same way we think about parties when considering entanglement distillation, network or transport layer protocols of the OSI model of the classical Internet. An intermediate scale quantum Internet example is a network where a source shares a GHZ state to all parties at each time-step, that they either use or verify. Our verification protocol can in this case be hidden in the assumptions of the network.

 One may wonder why we did not start by defining the multi-round and the verified GHZ sharing resources of the above section from the beginning. This is indeed the practical resource that one would like to use in larger protocols as it directly provides quantum states that are $\epsilon$-close to the GHZ state. This was based on the fact that our priority was not to define \emph{ad hoc} the most useful resource, but to succeed in modeling a resource that is as close as possible to the signals that will actually be sent by the parties when performing the protocol in real life, and use this resource in a composably secure way to obtain a practical multipartite entanglement verification resource, that of the multi-round resource.
 Our one-round resource captures the important parameters for composing the protocol in larger routines and it allows for modularity and a more precise understanding of what happens in the multi-round case. We will also see below that  dishonest behavior of a party already causes composability issues in the one-round case thus we get a better understanding of the issues by proving composability in this case. Moreover the box-shaped resource that we construct using AC (Fig.~\ref{fig:IdealHonestFilteredMEV}) is close to the black-box picture that we would like to have when thinking of the building blocks of the Quantum Internet.  Finally, we emphasize that this protocol only assumes classical communications between the parties and single-qubit local operations for each party, making it a good candidate for scalable application development. 

\subsection{Case of a malicious party}
\label{subsec:MaliciousParty}

When studying this problem, it is natural to think about the case of dishonest parties possibly controlling the source. If we assume that dishonest parties are trying to make the others accept a state that is not close the GHZ state, results from~\cite{MEVresistant} and~\cite{MEVexperimental} show that for one round of verification, the output bit $b_{out}$ depends on the minimal distance between the GHZ state and the shared state up to local operations on the part of the state held by the dishonest parties. This result holds even when the dishonest parties have complete control over the state generation resource. For this to hold, we have to assume that the Verifier is always honest and that the parties cannot influence the probability distributions of the oracles $\mathcal{O}_C$ and $\mathcal{O}_v$.

Yet as discussed in the first part of this paper, it seems that this protocol cannot be proven composable in the Abstract Cryptography framework when considering a dishonest party. Indeed one straightforward strategy for a dishonest party would be to make the protocol abort randomly, which would give false information about the source. Any dishonest party actually has complete control on the distribution of the concrete resource's output $b_{out}$ while the ideal resource's output is fixed by the distance with the GHZ state of the state given as input. Even if we add switches to our resource on which a simulator could act to make it abort (as custom is such cases), we could not reproduce the abort probability distribution of our concrete protocol in the ideal world. It seems impossible to find a simulator that emulates the interfaces a distinguisher has access to when removing one of the $\pi_i$. This can be seen in the AC framework by removing the converters corresponding to the dishonest parties and finding distinguishing attacks for every possible simulator. We would moreover need extra assumptions on the quantum registers and the access to the multiparty computation oracles $\mathcal{O}_C$ and $\mathcal{O}_v$ that seem unpractical in a near-term network.

However, our composability result comes on top of the security proof of~\cite{MEVresistant} meaning that our multiparty entanglement verification protocol is secure against possible coalition of dishonest parties and source trying to persuade others that they share a GHZ when they don't and composably secure against a malicious source. It does not limit the use of our protocol to the all honest case. No attack is known to make use of the repetition of the protocol that would alter the \textit{integrity} of the shared state more than simply repeating the attack described in ~\cite{MEVresistant}. On the other hand, the \textit{availability} of the resource can be compromised by dishonest behaviour in an unpredictable way. This sheds light on the pros and cons of using a game-based framework versus a composable framework. In the former we can restrict dishonest behaviour to specific attacks and get specific security properties while in the former we can only act on how powerful the class of distinguishers is but get more general security claims. By studying the protocol in different frameworks, we are able to take the best of all approaches and show different aspects of security that increase confidence in the protocols.

\subsection{Practical implementation in a near-term quantum network}
\label{subsec:PracticalImp}

To actually implement the protocol, one has to replace the resources in $\R$ with actual protocols or physical resources. Multiparty classical protocols should take the place of oracle calls, and have to be proved composable to securely construct $\R$ out of them and the quantum resources. An example of a protocol replacing calls to $\mathcal{O}_C$ is the random bit protocol explicited in~\cite{Anonymity} and~\cite{classAlgo}. Previous work~\cite{MEVresistant} shows that by choosing the probability of using the qubit for computation ($C=0$) to be $\frac{\epsilon^2}{4n\delta}$ for some $\delta>0 $, all honest parties have the guarantee that the probability the state used has distance at least $\epsilon$ from the correct
one is at most $\frac{1}{\delta}$. 

Qubit transportation should be taken care of by physical channels and link layer protocols that one has to study to see if they are equivalent to the quantum channel resource presented in this paper. As previously mentioned, any noisy channel can be modeled by a perfect channel in which a noisy state is given as input, and a noisy source can be modeled by a perfect source in which a classical description of a noisy state is given. The $\mathcal{SG}_n$ resource is designed as an attempt to capture what happens in the most general case when the protocol is performed in the lab where, at some point, a classical signal is sent to a quantum device that creates a state. Usually some information is accessible to the person controlling the device to check (for example by heralding photons) if the right state has been created. We suppose here that none of this information leaks from $\mathcal{SG}_n$ as it is the more restricting scenario. Moreover we don't restrict the source to create only $n$-qubit states, but merely enforce that it is able to create states up to this size. The proof holds even if the source creates bigger states and keeps part of it or sends it to a malicious party. $\mathcal{SG}_n$ is thus not meant to be realistic but to give an abstract embodiment of any source.  A photonic implementation of a loss-tolerant variation of the original protocol has been achieved with 4 parties~\cite{MEVexperimental}. This leads to expect near-term realization of our protocol, presenting all security properties as well as composability and modularity for use in bigger protocols.

Lastly, the quantum memory assumption can be removed by asking the parties to measure their bit directly after receiving it and flipping the outcome randomly depending on the input given by the Verifier. We would lose the security properties against a malicious party from~\cite{MEVresistant} that are based on the actual order of the inputs for each party. In our all honest setting, this would not matter so this protocol can actually be used in near-term architecture to securely check a source. Experimental realization of this protocol in a composable way is currently studied, which would allow to take this protocol as a concrete building block for applications in the quantum Internet. Whether this protocol should remain in the application layer or be hidden in some network or transport layer is still to be determined and will depend on future developments in quantum network architectures.

\section*{Acknowledgements}
We would like to thank Simon Neves, Léo Colisson, Atul Mantri, Anna Pappa, Damian Markham and Frédéric Grosshans for fruitful discussions. This project is part of the Quantum Internet Alliance and has received funding from the European Union’s Horizon 2020 research and innovation program under grant agreement No 820445. We also acknowledge support from the European Union through the Project ERC-2017-STG-758911 QUSCO, from QuantERA through the project QuantAlgo and from the ANR through the Project ANR-17-CE39-0005 quBIC. 

\section*{References}
\bibliographystyle{ieeetr}
\bibliography{article}

\end{document}